\documentclass[%
 reprint,
superscriptaddress,
nofootinbib,
 amsmath,amssymb,
 aps,
 pra,
]{revtex4-2}

\usepackage{graphicx}
\usepackage{dcolumn}
\usepackage{bm}
\usepackage{physics}
\usepackage[dvipsnames]{xcolor}
\usepackage{comment}
\usepackage{amsthm}
\usepackage{mathrsfs}
\usepackage{float}
\usepackage{dsfont}
\usepackage[colorlinks=true,bookmarks=false,citecolor=blue,urlcolor=blue,linkcolor=blue]{hyperref}
\defcitealias{Baker2020}{Nat. Phys. \textbf{17} 179 (2021)}	

\begin{document}

\preprint{APS/123-QED}

\title{Analytical results for laser models producing a beam with sub-Poissonian photon statistics and coherence scaling as the Heisenberg limit}

\author{L.~A.~Ostrowski}
 \email{l.ostrowski@griffith.edu.au}
\affiliation{%
 Centre for Quantum Dynamics, Griffith University, Yuggera Country, Brisbane, Queensland 4111, Australia
}%
\author{T.~J.~Baker}
\affiliation{%
 Centre for Quantum Dynamics, Griffith University, Yuggera Country, Brisbane, Queensland 4111, Australia
}%
\author{D.~W.~Berry}
\affiliation{
 Department of Physics and Astronomy, Macquarie University, Dharug Country, Sydney, New South Wales 2109, Australia
}%
\author{H.~M.~Wiseman}
 \email{h.wiseman@griffith.edu.au}
\affiliation{%
 Centre for Quantum Dynamics, Griffith University, Yuggera Country, Brisbane, Queensland 4111, Australia
}%

\date{\today}

\begin{abstract}

Recent advances in laser theory have demonstrated that a quantum enhancement is possible for the production of coherence $\mathfrak{C}$ by a continuous-wave laser device. Curiously, natural families of laser models that achieve Heisenberg-limited scaling for coherence produce the most coherence when the beam exhibits sub-Poissonian photon statistics. In this work, we provide an analytical treatment of those novel families of laser models by considering a parameter regime that permits a linearization. We characterize the dynamics of each laser system, and find that some of the intuitions from standard laser theory may be applied here. Specifically, the intracavity number dynamics are well-described as an Ornstein-Uhlenbeck process, while the intracavity phase dynamics are well-described in terms of a physically realizable ensemble of pure states, which evolve according to pure phase diffusion. Unlike a standard laser, however, we find that the pure states comprising the ensemble in the Heisenberg-limited lasers are substantially phase squeezed. From our dynamical analysis, we deduce various quantities of the beam for each laser family, including the first- and second-order Glauber coherence functions, intensity noise spectrum, Mandel-Q parameter and coherence $\mathfrak{C}$. In addition, inspired from these phase diffusion dynamics, we derive an upper bound on laser coherence $\mathfrak{C}\lesssim1.1156\mu^4$---which is tighter by a factor of $3/8$ when compared to that derived in [Baker \textit{et} al.,~\citetalias{Baker2020}]---by making one of the assumptions of that paper slightly stronger.

\end{abstract}

\maketitle


\section{Introduction}
\label{intro}

An essential feature of a laser beam that gives it widespread utility is its ability to act as a classical phase reference. A measure of this property has been proposed in Ref.~\cite{Baker2020}, with the quantity coined as ``laser coherence", denoted by $\mathfrak{C}$. This is defined as the mean number of photons within the maximally populated mode of a quantum field (within some frequency band, if required). In particular scenarios, this quantity can be expressed in a rather intuitive manner. For example, if we consider a light beam that is in a coherent state undergoing a process of pure phase diffusion---which is regarded an apt description of an ideal laser beam~\cite{Louisell1973,SSL_1974}---then one has $\mathfrak{C}=4\mathcal{N}/\ell$~\cite{Baker2020}. Here, $\mathcal{N}$ is the beam photon flux and $\ell$ is the diffusion rate. For this case, we can consider $\mathfrak{C}$ as roughly the number of photons emitted into the laser beam within a coherence time $\tau_{\rm coh} = 1/\ell$; i.e., the number of photons emitted into the beam that are mutually coherent.

With this measure of coherence, a fundamental question can be posed. That is, given $\mu$, the mean number of coherent excitations stored within a laser device, what are the ultimate bounds on $\mathfrak{C}$ in terms of that energy resource? The answer that one would obtain from an application of standard laser theory is that $\mathfrak{C}$ can be much greater than $\mu$, with the quantum limit $\mathfrak{C}_{\rm SQL} = O(\mu^2)$~\cite{Wiseman_1999}. This limit, however, arises as result of imposing particular assumptions about how a laser device gains and subsequently releases its energy to form a beam. Because these assumptions stem from conventional techniques of building a laser, rather than any fundamental aspect of quantum theory, this limit represents a standard quantum limit (SQL).

Recent works have challenged this notion by relaxing the common assumptions about how a laser device operates, and investigating the coherence properties of the ensuing laser beams~\cite{Baker2020,Liu2021,Ostrowski2022a,Ostrowski2022b}. Notably, in Ref.~\cite{Baker2020} an upper bound $\mathfrak{C}_{\rm HL}=O(\mu^4)$ was rigorously proven, based on a set of four natural assumptions for a laser device to satisfy. These assumptions constrain the beam properties to closely resemble those of an ideal laser beam state mentioned above, and require the phase information in the beam to proceed \textit{exclusively} from the coherent excitations $\mu$ stored within the device. Crucially, however, these assumptions are agnostic about the exact mechanisms by which the stored coherent excitations produce coherence in the beam. Moreover, that upper bound was shown to be achievable; a laser model saturating the $\mu^4$ scaling law for $\mathfrak{C}$ was shown to exist, while also satisfying the four assumptions. As such, $\mathfrak{C}_{\rm HL}=O(\mu^4)$ can be considered as an ultimate quantum limit, or Heisenberg limit (HL), for laser coherence, suggesting the possible advent of a new class of ultra-coherent beams of light.

The Heisenberg limit for $\mathfrak{C}$ has been generalized in Refs.~\cite{Ostrowski2022a,Ostrowski2022b}, whereby the upper bound $\mathfrak{C}_{\rm HL}=O(\mu^4)$ was derived using relaxed constraints on the beam properties. These relaxed constraints allow the beam to deviate further from a phase-diffusing coherent state, and encompass laser beams which can exhibit a significant degree of sub-Poissonianity in the photon statistics of the output field. The display of sub-Poissonian photon statistics within a laser beam can be seen as a signature of amplitude or number squeezing~\cite{Ralph_2004,Bachor_2019}, and such bright, nonclassical light finds numerous applications in quantum-optical technologies. Examples of these applications include atomic control~\cite{Goldberg_2020}, communication~\cite{Golubeva_2008,Hosseinidehaj_2022}, computation~\cite{Korolev_2019} and metrology~\cite{S_nchez_Mu_oz_2021}. The generalization presented in Refs.~\cite{Ostrowski2022a,Ostrowski2022b} thus facilitated the study of laser devices that produce beams exhibiting both Heisenberg-limited coherence, and photon number fluctuations below the shot-noise limit, which are two key properties that could find utility throughout precision technology. 

An additional curious insight was found in Refs.~\cite{Ostrowski2022a,Ostrowski2022b}. There, three families of laser models were developed, and each were shown to exhibit Heisenberg limited coherence, with $\mathfrak{C}$ saturating the $\mu^4$ scaling law. Two of those families also demonstrated sub-Poissonian beam photon statistics, which was quantified by the Mandel-$Q$ parameter for long counting intervals~\cite{Mandel79,Zou1990,Davidovich_1996}. From an investigation of those two sub-Poissonian laser families, it was found that an increase in the degree of sub-Poissonianity in the beam was positively correlated with an increase in its coherence. That is, the opposite of a trade-off was demonstrated between those two quantities for those particular families of laser models.

All studies to date regarding laser models that exhibit Heisenberg-limited coherence have been based either on numerical or heuristic analyses. In this work, we derive analytical results for a variety of these models, in order to deepen our intuition regarding these novel types of laser devices. In particular, we consider those three families developed in Refs.~\cite{Ostrowski2022a,Ostrowski2022b} mentioned above, and specialize to a parameter regime that permits a linearization. From our analysis, we obtain a clear picture of the dynamics for each system, and show that many of the intuitions from standard laser theory apply here. For example, we find that the evolution of the cavity state for our considered families of laser models resembles a pure state with its phase variable $\phi$ undergoing pure diffusion. This behavior is analogous to a stochastically-evolving coherent state in a standard laser cavity~\cite{Noh_2008,Carmichael_2010}; however, unlike a coherent state, the pure cavity state in the Heisenberg-limited case exhibits substantial phase squeezing. Additionally, we find that the intracavity number dynamics resembles an Ornstein-Uhlenbeck process, much like the behavior seen in sub-Poissonian laser models developed throughout the late 1980s and early 1990s~\cite{Haake_1989}. These insights permit straightforward formulas for many important physical properties of the beams to be obtained analytically. These include the coherence $\mathfrak{C}$, along with the Mandel-$Q$ parameter and intensity noise spectrum.

Motivated by the aforementioned phase diffusion dynamics, we also derive a tighter upper bound on the coherence, $\mathfrak{C}\lesssim1.1156\mu^4$ for $\mu\gg1$; the prefactor here being $3/8$ times that derived in Ref.~\cite{Baker2020}. (Note that we use the symbol $\lesssim$ in an asymptotic sense, where $f(x)\lesssim g(x)$ should be understood as ${\rm lim}_{x\to\infty}\{{\rm sup}_{x'\geq x}[f(x')/g(x')]\}\leq1$.) While the condition used on the beam to obtain this upper bound is more elegant than those used in previous proofs, it is slightly stricter. Regardless, we provide strong evidence that a laser model exists that both satisfies this stricter beam constraint and exhibits Heisenberg-limited coherence. This model is a special case of the families studied within this work.

This paper is structured as follows. In Section~\ref{models}, we review some of the physical concepts that are relevant to this study, before introducing the families of laser models that we will specifically consider. Here, we also provide a description of the parameter regime which we specialize to in order for our analytical treatment to proceed. In Section~\ref{number}, we investigate the cavity photon-number dynamics for our families of laser models, from which we evaluate properties relating to the beam photon-number statistics. In Section~\ref{phase}, we look at the phase dynamics for our families of laser models. Here, we are able to show that a physically realizable (PR) ensemble~\cite{Wiseman2001} of cavity states exist for each family, and that the dynamics of these ensembles is that of pure phase diffusion. Following this, we are able to evaluate the beam coherence $\mathfrak{C}$ for each family. The tighter upper bound for $\mathfrak{C}$ is proven in Section~\ref{tighter ub}. There, we also provide evidence that a laser model exists which satisfies the constraints used to derive that bound, which also exhibits Heisenberg-limited coherence. Section~\ref{conclusion} concludes this work and provides a brief discussion of future research directions.

\section{Heisenberg-Limited, Sub-Poissonain Laser Models}
\label{models}

\subsection{Coherence Properties of a Laser Beam}\label{sec:coherence_general}

Ideally, the output of a continuous-wave laser in free space may be considered as a one-dimensional bosonic beam that propagates at a fixed speed, with translationally invariant statistics. Such a beam can be described by a single-parameter field operator $\hat{b}(t)$, satisfying the commutation relation\footnote{While a one-dimensional field is considered here, a similar commutation relation can be derived without making a paraxial approximation~\cite{van_Enk_2004}.} $[\hat{b}(t),\hat{b}^\dagger(t')] = \delta(t-t')$\cite{Gardiner_1985}. The optical coherence properties of this beam can be characterized using the family of correlation functions introduced by Glauber~\cite{Glauber_1963}. This family of $n^{\rm th}$-order correlation functions are defined in terms of $2n$ field operators evaluated at different spatial coordinates $s_i$ of the field,
\begin{align}\label{Glauber_general}
    G^{(n)}(s_1,...,s_{2n}):=\langle\hat{b}^\dagger(s_1)...\hat{b}^\dagger(s_n)\hat{b}(s_{n+1})...\hat{b}(s_{2n})\rangle,
\end{align}
while the corresponding normalized forms are given as
\begin{align}\label{gn normal}
    \begin{split}
        g^{(n)}(s_1,...,s_{2n})&:=\frac{G^{(n)}(s_1,...,s_{2n})}{\Pi_{i=1}^{2n}\sqrt{G^{(1)}(s_i,s_i)}} \\ & =\frac{G^{(n)}(s_1,...,s_{2n})}{\mathcal{N}^n}.
    \end{split}
\end{align}
The second equality in~(\ref{gn normal}) follows from the translational invariance of the field, and we have defined the \textit{photon flux} in the beam as $\mathcal{N} := G^{(1)}(s_i,s_i)$, which has units of inverse time.

Throughout this work, the first- and second-order correlation functions are the subject of particular interest. The former, $G^{(1)}(s+t,s)$, is useful for quantifying the phase fluctuations in the field, and the coherence $\mathfrak{C}$ of a laser beam may be expressed in terms of this quantity. In particular, for a beam that exhibits translationally invariant statistics, the coherence, as defined in Ref.~\cite{Baker2020}, is proportional to the peak of the dimensionless power spectrum [i.e., the peak of the Fourier transform of $G^{(1)}(s+t,s)$]. Working in an optical frame rotating at this peak frequency, we thus have~\cite{Baker2020}
\begin{align}\label{coh g1}
    \mathfrak{C} = \int_{-\infty}^\infty dt\,G^{(1)}(s+t,s).
\end{align}

The second-order correlation function with the specific time ordering
\begin{align}\label{g2ps}
    g^{(2)}_{\rm ps}(t) := g^{(2)}(s,s+t,s+t,s)
\end{align}
is useful for characterizing the photon statistics of the beam (as indicated by the subscript ``ps"). The function $g_{\rm ps}(t)$ can be interpreted as the relative change in the likelihood of observing a subsequent photon at a time $t$ later than an initial beam photon detection. The utility of this correlation function can be highlighted by considering, for example, the Mandel-$Q$ parameter~\cite{Mandel79} for the beam over some time interval $[s,s+t)$, defined as
\begin{align}\label{mandelQ}
    Q_{s,t}:=\frac{\langle(\Delta\hat{n}_{s,t})^2\rangle - \langle\hat{n}_{s,t}\rangle}{\langle\hat{n}_{s,t}\rangle},
\end{align}
with the beam photon number operator in $[s,s+t)$ being $\hat{n}_{s,t} := \int_s^{s+t}dt'\,\hat{b}^\dagger(t')\hat{b}(t')$ and $\langle(\Delta \hat{c})^2\rangle=\langle \hat{c}^2 \rangle - \langle \hat{c}\rangle^2$ denoting the variance. Eq.~(\ref{mandelQ}) quantifies the degree of sub-Poissonianity in the field (occurring when $-1\leq Q_{s,t}<0$), which can be viewed as a measure of non-classicality~\cite{Davidovich_1996}. Again using the time-translational invariance of $\hat{b}(t)$, we can express the Mandel-$Q$ parameter as~\cite{Zou1990}
\begin{align}\label{Q_t}
    Q_{t} = \frac{\mathcal{N}}{t}\int_{-t}^{t}dt'(t-|t'|)\left[g_{\rm ps}^{(2)}(t')-1\right],
\end{align}
where we have dropped the redundant subscript $s$. For the types of nonclassical beams that are considered in this work, it will be the case that the function $[g^{(2)}_{\rm ps}(t)-1]$ is concave for $t>0$, and monotonically increases from a negative value at $t=0$ to zero in the long-time limit $t\to\infty$. For such fields, it follows from Eq.~(\ref{Q_t}) that maximal sub-Poissonianity is found in the limit of long counting intervals, with $Q_{t\to\infty}$.

Another way one can quantify the photon statistics of the beam is by analyzing its intensity noise spectrum. This is defined as
\begin{align}\label{spec intense}
    \begin{split}
        S_I(\omega) := \int_{-\infty}^\infty dte^{i\omega t}\Big[ & \langle\hat{b}^\dagger(s+t)\hat{b}(s+t)\hat{b}^\dagger(s)\hat{b}(s)\rangle \\ & - \langle\hat{b}^\dagger(s+t)\hat{b}(s+t)\rangle\langle\hat{b}^\dagger(s)\hat{b}(s)\rangle\Big] \\ & \hspace{-2.25cm} = \mathcal{N} + \mathcal{N}^2\int_{-\infty}^\infty dt\,{\rm cos}(\omega t)\left[g_{\rm ps}^{(2)}(t) - 1\right],
    \end{split}
\end{align}
where the equality on the second line also follows from the time-translational invariance of $\hat{b}(t)$. The first term on the the second line of Eq.~(\ref{spec intense}) is a result of the shot noise in the beam, and for a Poissonian beam [$g^{(2)}_{\rm ps}(t) = 1$] the spectrum~(\ref{spec intense}) will be equal to this value for all $\omega$. Alternatively, for a beam exhibiting some degree of sub-Poissonianity, where the intensity fluctuations of the beam are reduced below the shot noise limit, then one has $S_I(\omega)<\mathcal{N}$ for some range of $\omega$~\cite{Tuck_2006}. In this scenario, we will refer to such beams as being \textit{number squeezed}. Evidently, the value of the intensity noise spectrum~(\ref{spec intense}) is closely related to the long-time Mandel-$Q$ parameter; from Eqs.~(\ref{Q_t}) and (\ref{spec intense}), we find 
\begin{align}
    \lim_{t\to\infty}Q_{t} = \frac{S_{I}(\omega=0)}{\mathcal{N}}-1.
\end{align}

A laser is typically idealized as a device that produces a beam that is describable as a coherent state undergoing pure phase diffusion~\cite{Louisell1973,SSL_1974}. That is, the beam resides in an eigenstate $\ket{\beta(t)}$ of $\hat{b}(t)$, so that, for all $t$,
\begin{align}\label{ideal laser}
    \hat{b}(t)\ket{\beta(t)} = \beta(t)\ket{\beta(t)},
\end{align}
with the eigenvalue $\beta(t) = \sqrt{\mathcal N}e^{i\sqrt{\ell}W(t)}$. Here, $\ell$ represents the rate of phase diffusion and $W(t)$ is a Wiener process. The Glauber coherence functions for such a state can be readily evaluated, giving $G^{(1)}(s+t,s) = \mathcal{N}e^{-\ell|t|/2}$ and $g^{(2)}_{\rm ps}(t) = 1$. The former expression implies a Lorentzian power spectrum with a linewidth (i.e., the full width at half maximum) of $\ell$, and subsequently, from Eq.~(\ref{coh g1}), a straightforward formula for the coherence, $\mathfrak{C} = 4\mathcal{N}/\ell$, as claimed in the introduction. The latter expression, $g^{(2)}_{\rm ps}(t) = 1$, implies Poissonian beam photon statistics, where there is no correlation between the detection of beam photons.

If one restricts $\hat{b}(t)$ to closely exhibit these ideal laser beam properties, then the Heisenberg limit for the coherence in the beam is upper bounded by $\mathfrak{C}\lesssim2.9748\mu^4$ for $\mu\gg1$~\cite{Baker2020}. Here, we have the energy resource $\mu:=\langle\hat{n}_{\rm c}\rangle_{\rm ss}$, where the subscript ``ss" specifies a steady state average, which is assumed to exist and be unique, and $\hat{n}_{\rm c}$ is the generator of phase shifts on the state of the laser device $\rho_{\rm c}$. That is, a phase shift of $\theta$ on $\rho_{\rm c}$ is achieved by the action of the superoperator $\mathcal{U}^\theta_{\rm c}(\bullet):=e^{i\theta\hat{n}_{\rm c}}\bullet e^{-i\theta\hat{n}_{\rm c}}$. In this case, the generator $\hat{n}_{\rm c}$ has nonnegative integer eignevalues, such that $\mathcal{U}^{\theta+2\pi}_{\rm c}(\bullet) = \mathcal{U}^\theta_{\rm c}(\bullet)$, and $\mu$ may be interpreted as an excitation number. Often (but not necessarily), $\mu$ represents mean photon number within a single mode of a laser cavity. Indeed, laser models which achieve the upper bound scaling of $\mathfrak{C}=\Theta(\mu^4)$ have been shown exist, in principle, with the energy resource $\mu$ stored in a single mode~\cite{Baker2020,Ostrowski2022a,Ostrowski2022b}. The key components of these laser models that enable them to achieve a coherence enhancement are the specific quantum interactions between the device and its environment. These interactions govern the manner in which the device gains incoherent excitations from a source, and releases its stored coherent excitations to form a beam. They are designed in such a way that the beam exhibits properties that are very close to the ideal beam state described in Eq.~(\ref{ideal laser}), yet they preserve phase information within the device more effectively compared to a conventional laser.

By letting go of the strict notion that a laser beam closely approximates the ideal beam state $\ket{\beta(t)}$ in Eq.~(\ref{ideal laser}), Heisenberg limits for other, more exotic, types of laser beams may also be considered. For example, \textit{sub-Poissonian lasers} produce a highly coherent beam of radiation and exhibit photon number fluctuations below the shot noise limit~\cite{golubev1984,Machida_1987,Ralph_2004,Walls_Milburn_2008}. Recently, in Refs.~\cite{Ostrowski2022a,Ostrowski2022b}, the Heisenberg limit $\mathfrak{C}=\Theta(\mu^4)$ was generalized to account for more deviation in the first- and second-order Glauber coherence functions. While these more relaxed constraints still required the beam to \textit{passably approximate} the properties of an ideal laser beam, 
they also permit deviations in $g_{\rm ps}^{(2)}(t)$ that are significant enough to encompass beams with maximal sub-Poissonianity in the long time limit; i.e., $Q_{t\to\infty}=-1$. Moreover, in those works, two families of sub-Poissonian laser models that saturate the Heisenberg limit scaling law $\mathfrak{C}=\Theta(\mu^4)$ were shown to exist, both of which can produce a beam with a higher coherence (i.e., a constant factor enhancement) over the highest performing Poissonian laser model currently known~\cite{Baker2020}. Both of these sub-Poissonian families of laser models, which are the subject of analytical treatment in this work, are now introduced.

\subsection{The $p,\lambda$-Family of Laser Models}

The first family of laser models that we consider can be described with a master equation in Lindblad form
\begin{align}\label{lambda_master}
    \begin{split}
        \frac{d\rho}{dt} & = \mathcal{L}^{(p,\lambda)}\rho = r\mathcal{N}\left(\mathcal{D}[\hat{G}^{(p,\lambda)}] + \mathcal{D}[\hat{L}^{(p,\lambda)}]\right)\rho.
    \end{split}
\end{align}
Here, $\rho$ is the state matrix of a non-degenerate, $D$-dimensional quantum system, which we refer to as a ``cavity", $\mathcal{N}$ is the photon flux, $\mathcal{D}[\hat{c}]\bullet:=\hat{c}\bullet\hat{c}^\dagger - (1/2)(\hat{c}^\dagger\hat{c}\bullet+\bullet\hat{c}^\dagger\hat{c})$ is the Lindblad superoperator, and $\hat{G}^{(p,\lambda)}$ and $\hat{L}^{(p,\lambda)}$ are respectively the operators characterizing the gain and loss of excitations to and from the cavity. The non-zero elements of these two operators are defined in the number basis of the cavity as
\begin{subequations}
    \begin{align}\label{gainop}
        \begin{split}
            \tilde{G}_n(p,\lambda) & := \bra{n}\hat{G}^{(p,\lambda)}\ket{n-1} = \left(\frac{\sin\left(\pi\frac{n+1}{D+1}\right)}{\sin\left(\pi\frac{n}{D+1}\right)}\right)^{\frac{p\lambda}{2}},
        \end{split}
    \end{align}
    \begin{align}\label{lossop}
        \begin{split}
            \tilde{L}_n(p,\lambda) & := \bra{n-1}\hat{L}^{(p,\lambda)}\ket{n} = \left(\frac{\sin\left(\pi\frac{n}{D+1}\right)}{\sin\left(\pi\frac{n+1}{D+1}\right)}\right)^{\frac{p(1-\lambda)}{2}},
        \end{split}
    \end{align}
\end{subequations}
where $0<n<D$.

Here, the parameter $p$ modifies the sharpness of the steady state cavity distribution $\rho_{\rm ss}$, as
\begin{align}\label{ss}
    \begin{split}
        \rho_{\rm ss} & \propto \sum_{n=0}^{D-1}\sin^p\left(\pi\frac{n+1}{D+1}\right)\ketbra{n}{n},
    \end{split}
\end{align}
which has mean
\begin{align}
    \mu={\rm Tr}\left(\hat{a}^\dagger\hat{a}\rho_{\rm ss}\right) = \frac{D-1}{2}.
\end{align}
By contrast, $\lambda\in\mathbb{R}$ modifies the overall ``flatness" of the coefficients of the gain and loss operators, while preserving the steady state distribution~(\ref{ss}). For example, $\lambda=0$ sets a flat gain with $\bra{n}\hat{G}^{(p,0)}\ket{n-1}=1$, while $\lambda=1$ instead sets a flat loss, with $\bra{n}\hat{L}^{(p,-1)}\ket{n-1}=1$. We note that the factor $r$ in Eq.~(\ref{lambda_master}) is defined to ensure $r{\rm Tr}\{\hat{L}^{(p,\lambda)\dagger}\hat{L}^{(p,\lambda)}\rho_{\rm ss}\} = 1$. This factor differs from unity by terms that are $O(D^{-2})$ when $\lambda\neq0$; this deviation from unity can be neglected as our analysis exclusively considers the limit $D\to\infty$.

Throughout this work, we assume that the term giving rise to the formation of the output beam in~(\ref{lambda_master}), $\mathcal{D}[\hat{L}^{(p,\lambda)}]\bullet$, arises from quantum vacuum white noise coupling~\cite{wiseman_milburn_2009}. This means that $\mathcal{N}$ can indeed be interpreted as the photon flux from the cavity and, more generally, that the Glauber coherence functions of the beam can be expressed in terms of the loss operator~\cite{Gardiner_1985}; e.g.,
\begin{subequations}
    \begin{align}\label{G1L}
        G^{(1)}(s+t,s) = \mathcal{N}{\rm Tr}\left\{\hat{L}^{(p,\lambda)\dagger}e^{\mathcal{L}^{(p,\lambda)}t}\left(\hat{L}^{(p,\lambda)}\rho_{ss}\right)\right\},
    \end{align}
    \begin{align}\label{g2L}
        g_{\rm ps}^{(2)}(t) = {\rm Tr}\left\{\hat{L}^{(p,\lambda)\dagger}\hat{L}^{(p,\lambda)}e^{\mathcal{L}^{(p,\lambda)}t}\left(\hat{L}^{(p,\lambda)}\rho_{ss}\hat{L}^{(p,\lambda)\dagger}\right)\right\}.
    \end{align}
\end{subequations}

For the $p,\lambda$-Family of laser models, numerical evaluations of $\mathfrak{C}$ have shown that, in the limit of $D\to\infty$, Heisenberg limited coherence is achieved for $p>3$~\cite{Ostrowski2022a}. Within this regime, it was found that $\mathfrak{C}^{(p,\lambda)} \sim \mathfrak{a}(p,\lambda)\mu^4$ for some positive function $\mathfrak{a}(p,\lambda)$. The maximum of that function is independent in the two variables $\lambda$ and $p$. That is, for fixed $\lambda$, $\mathfrak{C}^{(p,\lambda)}$ is maximal at $p\approx4.15$, and for fixed $p$, $\mathfrak{C}^{(p,\lambda)}$ is maximal at $\lambda=0.5$. Furthermore, the photon statistics were seen to depend only on the value of $\lambda$, not $p$. Specifically, the Mandel-$Q$ parameter for long counting intervals was shown to reduce to $Q_{t\to\infty}=-0.5$ (the minimum value attained for this family) with the optimal choice of $\lambda=0.5$. Contrary to common intuitions~\cite{Bergou_1989b,Wiseman_1991,Wiseman1993}, the value of $\mathfrak{C}^{(p,\lambda)}$ is greater the more sub-Poissonian the beam is, for $p$ fixed, and the maximum value for $\mathfrak{C}^{(p,\lambda)}$ is twice as large as that obtained in the Poissonian case ($Q_{t\to\infty}=0$, when $\lambda\in\{0,1\}$).

The open quantum dynamics described by the master equation~(\ref{lambda_master}) are notably exotic, and realizing the necessary interactions between the laser mode and its environment in practice would require a high degree of engineering. Detailed proposals for implementing these dynamics using near-term technology are beyond the scope of this article. However, we note that experimental precedent exists for realizing such interactions on the platform of circuit quantum-electrodynamics. In Ref.~\cite{Li_2024}, for example, a Lindblad term akin to that of the gain term in Eq.~(\ref{lambda_master}) was implemented using reservoir engineering techniques in a three-dimensional cavity dispersively coupled an ancillary Transmon qubit and readout resonator (it was also pointed out in that work that desirable non-linear loss could be realised with minor modifications to that setup). Although the original goal of that experiment was to stabilize Fock states within the microwave cavity, the underlying mechanisms could be repurposed to realize the laser dynamics described by the $p,\lambda$-family of laser models. Investigations are ongoing to assess the feasibility of this approach to implement beyond-SQL laser dynamics~\cite{Ostrowski_2025}.

\subsection{The $p,q$-Family of Laser Models}

As just presented, the $p,\lambda$-family of laser models can achieve a value of $Q_{t\to\infty}=-0.5$ at best, which corresponds to a $50$\% reduction in the photon noise within the beam over long counting intervals. It is, however, possible to define a family of laser models that exhibits Heisenberg-limited coherence, while achieving a $100$\% reduction in the beam photon noise (i.e., a Mandel-$Q$ parameter $Q_{t\to\infty}=-1$). This can be done by modifying the statistics by which excitations are gained by the laser, as is done with the $p,q$-family of laser models~\cite{Ostrowski2022a,Ostrowski2022b}. This is the second family considered in this work and it can be described with the following master equation 
\begin{align}\raisetag{2\baselineskip}\label{q_master}
    \begin{split}
        \frac{d\rho}{dt} & = \mathcal{L}^{(p,q)}\rho \\ & =\mathcal{N}\left(\mathcal{D}[\hat{G}^{(p,0)}] + \frac{q}{2}\mathcal{D}[\hat{G}^{(p,0)}]^2 + r'\mathcal{D}[\hat{L}^{(p,-q/2)}]\right)\rho.
    \end{split}
\end{align}
Here, the gain and loss operators are the same as those defined in Eqs.~(\ref{gainop})~and~(\ref{lossop}), where we have fixed a flat gain with $\lambda=0$, while making the substitution $\lambda\to-q/2$ with $q\in(-1,\infty)$ for the loss operator. In the limit $D\to\infty$ with $p$ finite, this system yields the same steady state~(\ref{ss}) as for the $p,\lambda$-family. Here, we require a factor $r'=1+O(D^{-2})$ multiplying the loss term to ensure $r'{\rm Tr}\{\hat{L}^{(p,-q/2)\dagger}\hat{L}^{(p,-q/2)}\rho_{\rm ss}\} = 1$, but like the $p,\lambda$-family of laser models, we can set $r'=1$ in the large-$D$ limit. The parameter $q$ may be interpreted as the Mandel-$Q$ parameter of the (Gaussian) pumping process. Setting $q=0$ means that the number of excitations gained by the system in a time interval $\Delta t$ follows a Poissonian distribution. Alternatively, setting $-1\leq q<0$ gives a sub-Poissonian pumping process, where the variance in the number of excitations gained in $\Delta t$ is below the mean. In the extreme case of $q\to-1$, this variance approaches zero. The loss term in Eq.~(\ref{q_master}) is also assumed to arise from quantum vacuum white noise coupling, such that analogous expressions to~(\ref{G1L})~and~(\ref{g2L}) are obtained for this family.

Eq.~(\ref{q_master}) is only an approximation to a regularly-pumped laser system, being a Markovian equation that is attempting to describe a non-Markovian (for $q\neq0$) pumping process. In Ref.~\cite{Ostrowski2022a}, this approximate model was derived by considering the repeated action of a completely-positive, trace-preserving (CPTP) map representing a single round of gain within the laser, and is valid in the regime where $1\ll\mathcal{N}\Delta t\ll \mu$ and $3<p\ll\mu^2$. This represents the case where the number of excitations gained in $\Delta t$ is large, for which $q$ characterizes the gain statistics, but the overall change in the laser state caused by those repeated gain events is small. Such a regime allows for the stepwise gain events to be approximated as a continuous process. The form of this equation may appear unfamiliar due to the squared superoperator $\mathcal{D}[\hat{G}^{(p,0)}]^2$, and because of this term, it is not generally expressible in Lindblad form. However, the use of equations such as this to describe sub-Poissonian pumping in lasers is standard (see Refs.~\cite{Haake_1989,Bergou_1989a,Bergou_1989b,Wiseman1993}, for example), and the results generated from this model are reasonable within the parameter regime we focus on in this article. Further details about this model and the regime of its validity can be found in Appendix~\ref{rp_apx}.

Much like the $p,\lambda$-family, numerical calculations have shown that Heisenberg-limited coherence is achieved for the $p,q$-family in the limit $D\to\infty$ and $p>3$, where $\mathfrak{C}^{(p,q)}\sim\mathfrak{b}(p,q)\mu^4$ for some positive function $\mathfrak{b}(p,q)$. The maximum of this function is likewise independent in the two variables $q$ and $p$, such that, for fixed $q$, $\mathfrak{C}^{(p,q)}$ is maximal at $p\approx4.15$, and for fixed $p$, $\mathfrak{C}^{(p,\lambda)}$ is maximal at $q=-1$. The photon statistics of the laser beam depend only on $q$, where the long-time Mandel-$Q$ parameter of the beam simply mirrors that of the pump, with $Q_{t\to \infty}=q$. This is a result that can be expected by following the intuitions from early works concerning regularly pumped laser models~\cite{Yamamoto1992}. However, the coherence correspondingly increased when the Mandel-$Q$ parameter of the beam decreased, in a similar manner to what was observed for the $p,\lambda$-family. Notably, when the absolute minimum of $Q_{t\to \infty}=-1$ was obtained for the $p,q$-family, the coherence was four times as large as it was in the Poissonian case ($Q_{t\to \infty}=0$, when $q=0$).

\subsection{The Linearized Regime}

This subsection specifies key details of the parameter regime for the two families of laser models, which we refer to as the \textit{linearized regime}, and restrict our attention to for the remainder of this work. While it has been noted that optimal coherence is achieved for these families in the limit $\mu\to\infty$ with the relatively low value of $p\approx4.15$, the highly nonlinear forms of the operators characterizing the laser models for this choice of parameter value make analytical treatments intractable. Instead, in the linearized regime we consider moderately large values of $p$, such that the matrix elements of the operators $\hat{G}^{(p,\lambda)\dagger}\hat{G}^{(p,\lambda)}$ and $\hat{L}^{(p,\lambda)\dagger}\hat{L}^{(p,\lambda)}$ in the number basis of the cavity can be linearized. This will lead to simplified dynamics and hence make an analytical treatment possible, while preserving Heisenberg-limited scaling for $\mathfrak{C}$. The consideration of this regime requires two criteria to be met. First, we require that $1\ll p\ll\mu^2$, and second we require $\partial p/\partial\mu = 0$. The first criterion allows for the linearization to be made, while the second ensures that Heisenberg-limited scaling for $\mathfrak{C}$ is preserved for each family of laser models.

We explicitly show this linearization procedure for the operator $\hat{L}^{(p,\lambda)\dagger}\hat{L}^{(p,\lambda)}$, while an analogous procedure follows for the operator $\hat{G}^{(p,\lambda)\dagger}\hat{G}^{(p,\lambda)}$. This begins by writing down its non-zero coefficients in the number basis of the cavity
\begin{align}\label{Lsq}
    \begin{split}
        \left(\tilde{L}_n(p,\lambda)\right)^2 & = \bra{n}\hat{L}^{(p,\lambda)\dagger}\hat{L}^{(p,\lambda)}\ket{n} \\ & = \left(\frac{\sin\left(\pi\frac{n}{D+1}\right)}{\sin\left(\pi\frac{n+1}{D+1}\right)}\right)^{p(1-\lambda)}.
    \end{split}
\end{align}
Upon applying the elementary trigonometric identity $\sin(A\pm B) = \sin(A)\cos(B)\pm\cos(A)\sin(B)$ and taking $D\gg1$, Eq.~(\ref{Lsq}) can approximated as
\begin{align}\label{Lsq2}
    \left(\tilde{L}_n(p,\lambda)\right)^2 \approx \left[1 - \frac{\pi}{D+1}\cot\left(\pi\frac{n+1}{D+1}\right)\right]^{p(1-\lambda)}.
\end{align}
From here, it can be recognized that by considering moderately large values of $p$, the available cavity states would be well-localized about the mid-point $\mu=(D-1)/2$ of the interval for $n$. Thus, by considering the regime where $1\ll p\ll \mu^2$, we express the cotangent function in~(\ref{Lsq2}) as a Taylor series expanded about this mid-point, and further apply a binomial expansion to approximate the exponent $p(1-\lambda)$. Keeping terms to leading order in $p/\mu^2$ gives the desired linearized expression
\begin{align}\label{loss_lin}
    \left(\tilde{L}_n(p,\lambda)\right)^2 \approx 1 + \frac{\pi^2p(1-\lambda)}{(D+1)^2}\left(n - \frac{(D-1)}{2}\right).
\end{align}
Following similar arguments, the linearized expression for the coefficients of $\hat{G}^{(p,\lambda)\dagger}\hat{G}^{(p,\lambda)}$ are
\begin{align}
    \begin{split}\label{gain_lin}
        \left(\tilde{G}_{n+1}(p,\lambda)\right)^2 & := \bra{n}\hat{G}^{(p,\lambda)\dagger}\hat{G}^{(p,\lambda)}\ket{n} \\ & \approx 1 -\frac{\pi^2p\lambda}{(D+1)^2}\left(n - \frac{(D-1)}{2}\right).
    \end{split}
\end{align}
These approximations are visualized in Fig.~\ref{fig:linear_visual}.

\begin{figure}[H]
\includegraphics[width=1.0\columnwidth]{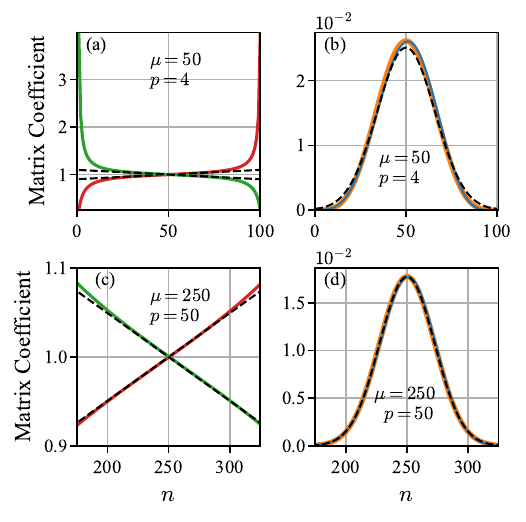}
\caption{\label{fig:linear_visual} (a): Diagonal matrix coefficients of the operators $\hat{L}^{(p,\lambda)\dagger}\hat{L}^{(p,\lambda)}$ (red) and $\hat{G}^{(p,\lambda)\dagger}\hat{G}^{(p,\lambda)}$ (green) in the number basis of the laser cavity for $\mu=50$, $p = 4$, and $\lambda = 0.5$. Linearized versions of these matrix elements, as defined in Eqs.~(\ref{Lsq2})~and~(\ref{gain_lin}), are displayed as dashed black lines. (b): Diagonal elements of the cavity steady state for the $p,\lambda$-family (orange) and $p,q$-family (blue) of laser models in the number basis. The Gaussian approximation of these distributions, as given in Eq.~(\ref{ss_lin}), is given by the dashed black curve. Panels (c) and (d) depict the same as that in (a) and (b), respectively, but for parameter value choices that more closely satisfy the linearized regime. Note the restricted range of $n$ plotted in these latter cases.}
\end{figure}

In the linearized regime, the cavity distribution for both families of laser models are approximately Gaussian. That is, for both families, we have
\begin{align}\label{ss_lin}
    \begin{split}
        & \rho_{\rm ss} \approx \sum_{n=0}^{D-1}\sqrt{\frac{k}{2\pi}}\exp\left[\frac{-k(n-\mu)^2}{2}\right]\ketbra{n}{n}; \quad k:=\frac{\pi^2p}{4\mu^2},
    \end{split}
\end{align}
regardless of the choices of $\lambda$ or $q$ for the respective families. Eq.~(\ref{ss_lin}) is compared with the steady state cavity distributions obtained numerically from Eq.~(\ref{lambda_master})~and~(\ref{q_master}) in Fig.~\ref{fig:linear_visual}.

With the linearized expressions for the gain and loss operators, we define two master equations. These equations approximately describe the dynamics of the $p,\lambda$- and $p,q$-families of laser models within the linearized regime; respectively these are
\begin{subequations}
    \begin{align}\label{lam_mast_lin}
        \begin{split}
            \frac{\dot{\rho}_{m,n}}{\mathcal{N}} = & G_{m}(k,\lambda)G_{n}(k,\lambda)\rho_{m-1,n-1} \\ & -\frac{1}{2}\big([G_{m+1}(k,\lambda)]^2 + [G_{n+1}(k,\lambda)]^2\big)\rho_{m,n} \\ & + L_{m+1}(k,\lambda)L_{n+1}(k,\lambda)\rho_{m+1,n+1} \\ & - \frac{1}{2}\left([L_{m}(k,\lambda)]^2 + [L_{n}(k,\lambda)]^2\right)\rho_{m,n}.
        \end{split}
    \end{align}
    \begin{align}\label{q_mast_lin}
        \begin{split}
            \frac{\dot{\rho}_{m,n}}{\mathcal{N}} = & \frac{q}{2}\rho_{m-2,n-2} + (1-q)\rho_{m-1,n-1} - (1-q/2)\rho_{m,n} \\ & + L_{m+1}(k,-q/2)L_{n+1}(k,-q/2)\rho_{m+1,n+1} \\ & - \frac{1}{2}\left([L_{m}(k,-q/2)]^2 + [L_{n}(k,-q/2)]^2\right)\rho_{m,n}.
        \end{split}
    \end{align}
\end{subequations}
Here, $\rho_{m,n} = \bra{m}\rho\ket{n}$, and we have defined the functions
\begin{subequations}
    \begin{align}\label{loss_coeff_linear_general}
        L_{n}(k,x) := \sqrt{1 + k(1-x)(n-\mu)},
    \end{align}
    \begin{align}
        G_{n+1}(k,x) := \sqrt{1 - kx(n-\mu)},
    \end{align}
\end{subequations}
which have the same form as Eqs.~(\ref{loss_lin})~and~(\ref{gain_lin}), but are expressed in terms of $k=\pi^2 p/4\mu^2$ and $\mu = (D-1)/2$. In terms of these parameters, $1/\mu^2\ll k\ll 1$ (which, ignoring a factor of $\pi^2/4$, essentially reproduces $1\ll p\ll\mu^2$) defines the regime for which Eqs.~(\ref{lam_mast_lin})~and~(\ref{q_mast_lin}) respectively approximate the original master equations~(\ref{lambda_master})~and~(\ref{q_master}). 

\section{Number Dynamics and Beam Photon Statistics}
\label{number}

In this section, we analyze the photon statistics of the beam produced by our two families of Heisenberg-limited laser models in the linearized regime. We confirm that sub-Poissonian beam photon statistics arise within both families by computing their long-time Mandel-$Q$ parameter, as well as the intensity noise spectrum. Both of these quantities follow from an analysis of the second order correlation function $g^{(2)}_{\rm ps}(t)$.

\subsection{Incoherent Cavity Dynamics}

For the purposes of this analysis, it is sufficient to consider the population dynamics of both families. Using the linearized Equations~(\ref{lam_mast_lin})~and~(\ref{q_mast_lin}), these dynamics are approximated by the following equations
\begin{subequations}
    \begin{align}\label{lam_num}
        \begin{split}
            \frac{\dot{\rho}_{n,n}}{\mathcal{N}} =  & \left[G_n(k,\lambda)\right]^2\rho_{n-1,n-1} - \left[G_{n+1}(k,\lambda)\right]^2\rho_{n,n} \\ & + \left[L_{n+1}(k,\lambda)\right]^2\rho_{n+1,n+1} - \left[L_n(k,\lambda)\right]^2\rho_{n,n},
        \end{split}
    \end{align}
    \begin{align}\raisetag{\baselineskip}\label{q_num}
        \begin{split}
            \frac{\dot{\rho}_{n,n}}{\mathcal{N}} = & (1-q)\rho_{n-1,n-1}-(1-q/2)\rho_{n,n}+\frac{q}{2}\rho_{n-2,n-2} \\ & + \left[L_{n+1}(k,-q/2)\right]^2\rho_{n+1,n+1} \\ & - \left[L_{n}(k,-q/2)\right]^2\rho_{n,n},
        \end{split}
    \end{align}
\end{subequations}
respectively for the $p,\lambda$-  and $p,q$-families.

Given that $\mu$ is large, it is possible to take the continuum limit, where we treat the coefficients of the gain and loss operators, and the cavity distribution as continuous functions of $n$. That is, $[G_n(k,x)]^2\to[G(n,k,x)]^2$, $[L_n(k,x)]^2\to[L(n,k,x)]^2$ and $\rho_{n,n}\to p(n,t)$. Then performing a second-order Taylor series expansion of the terms that include $p(n\pm1,t)$ and $p(n-2,t)$ in the corresponding equations, we can convert Eqs.~(\ref{lam_num})~and~(\ref{q_num}) into Fokker-Planck equations in terms of continuous probability distributions. Hence, for the $p,\lambda$-family we have
\begin{align}\raisetag{\baselineskip}\label{FPlam}
    \begin{split}
        \frac{\dot{p}(n,t)}{\mathcal{N}} & \approx \frac{\partial}{\partial n}\big(\big{\{}[L(n,k,\lambda)]^2 \\ & \hspace{1.5cm} - [G(n+1,k,\lambda)]^2\big{\}}p(n,t)\big) \\ & \hspace{.4cm} + \frac{1}{2}\frac{\partial^2}{\partial n^2}\big(\big{\{}[L(n,k,\lambda)]^2 \\ & \hspace{2.25cm} + [G(n+1,k,\lambda)]^2\big{\}}p(n,t)\big) \\ & \approx\frac{\partial}{\partial n}\left[k(n-\mu)p(n,t)\right] + \frac{\partial^2}{\partial n^2}p(n,t),
    \end{split}
\end{align}
wherein moving to the final line, we have kept the leading order term in $k$ within the second derivative. Similarly, for the $p,q$-family we have
\begin{align}\raisetag{\baselineskip}\label{FPq}
    \begin{split}
        \frac{\dot{p}(n,t)}{\mathcal{N}} & \approx \frac{\partial}{\partial n}\big(\big{\{}[L(n,k,-q/2)]^2 - 1\big{\}}p(n,t)\big) \\ & \hspace{.4cm} + \frac{1}{2}\frac{\partial^2}{\partial n^2}\big(\big{\{}1+q + [L(n,k,-q/2)]^2\big{\}}p(n,t)\big) \\ & \approx\frac{\partial}{\partial n}\left[k(1+q/2)(n-\mu)p(n,t)\right] \\ & \hspace{0.4cm} + (1+q/2)\frac{\partial^2}{\partial n^2}p(n,t).
    \end{split}
\end{align}

The last line of Eqs.~(\ref{FPlam})~and~(\ref{FPq}) are both examples of an Ornstein-Uhlenbck process. These dynamics are commonly seen throughout quantum optical systems~\cite{Walls_Milburn_2008}, with the number dynamics of a regularly pumped laser with linear output coupling being a quintessential example~\cite{Haake_1989}. Let us consider the general form of these equations, with drift and diffusion coefficients $\xi$ and $\chi$
\begin{align}
    \dot p(n,t) = \frac{\partial}{\partial n}\left[\xi(n-\mu)p(n,t)\right] + \chi\frac{\partial^2}{\partial n^2}p(n,t).
\end{align}
The solution to this equation is Gaussian, and can be readily obtained using the method of characteristics; for $t>0$, we find
\begin{align}\label{solution_cts}
    \begin{split}
        p(&m,t|n,0) = \\ & \frac{1}{\sqrt{2\pi\sigma(t)^2}}\exp{\frac{-\left[m-\mu - (n-\mu)e^{-\xi t}\right]^2}{2\sigma(t)^2}}.
    \end{split}
\end{align}
Here, we have the variance $\sigma(t)^2=\chi(1-e^{-2\xi t})/\xi$ and have also considered an initially sharp distribution $p(m,0|n,0)=\delta(m-n)$. Eq.~(\ref{solution_cts}) relaxes to a steady state with a mean $\mu$ and a variance $\sigma(t\to\infty)^2 = \chi/\xi$. Upon comparing the solutions that pertain to the two families of laser models, the key difference is the photon correlation time, which is $1/\xi\to 1/\mathcal{N}k$ for the $p,\lambda$-family and $1/\xi\to 1/\mathcal{N}k(1+q/2)$ for the $p,q$-family.

\subsection{The Second-Order Glauber Coherence Function}\label{g2_calculation}

Based on the general solution~(\ref{solution_cts}) for the incoherent dynamics of the two families of laser models, we now compute their second-order coherence functions, given by Eq.~(\ref{g2ps}). Considering the $p,\lambda$-family, we have the exact expression
\begin{align}
    \begin{split}\raisetag{3\baselineskip}\label{g2lam_exact}
        g^{(2)}_{{\rm ps},\lambda}(t) & = \sum_m[\tilde{L}_m(p,\lambda)]^2 \\ & \hspace{1cm} \bra{m}e^{\mathcal{L}^{(p,\lambda)}t}[\hat{L}^{(p,\lambda)}\rho_{\rm ss}\hat{L}^{(p,\lambda)\dagger}]\ket{m} \\ & = \sum_{m,n}[\tilde{L}_m(p,\lambda)]^2[\tilde{L}_{n+1}(p,\lambda)]^2 \\ & \hspace{0.75cm} \times\bra{n+1}\rho_{\rm ss}\ket{n+1}\bra{m}e^{\mathcal{L}^{(p,\lambda)}t}\left(\ketbra{n}{n}\right)\ket{m},
    \end{split}
\end{align}
and similarly, for the $p,q$-family, we have
\begin{align}
    \begin{split}\raisetag{2\baselineskip}\label{g2q_exact}
        g^{(2)}_{{\rm ps},q}(t) = \sum_{m,n} & [\tilde{L}_m(p,-q/2)]^2[\tilde{L}_{n+1}(p,-q/2)]^2 \\ & \hspace{-0.25cm} \times\bra{n+1}\rho_{\rm ss}\ket{n+1}\bra{m}e^{\mathcal{L}^{(p,q)}t}\left(\ketbra{n}{n}\right)\ket{m}.
    \end{split}
\end{align}

Working in the linearized regime and taking the continuum limit once again, Eqs.~(\ref{g2lam_exact})~and~(\ref{g2q_exact}) can both be expressed using the following general form
\begin{align}\label{g2_generalcts}
    \begin{split}\raisetag{1\baselineskip}
        g^{(2)}_{\rm ps}(t) \approx \int_0^\infty\int_0^\infty dm\,dn & [L(m,k,x)]^2[L(n+1, k, x, \mu)]^2 \\ & \times p(n+1)p(m,t|n,0),
    \end{split}
\end{align}
where $L(m,k,x)$ is the continuum form of Eq.~(\ref{loss_coeff_linear_general}), $p(m,t|n,0)$ is given in Eq.~(\ref{solution_cts}), and $p(m) = p(m,t\to\infty|n,0)$. Eq.~(\ref{g2_generalcts}) can be evaluated; we find
\begin{align}\label{g2_generalcts2}
    \begin{split}
        g^{(2)}_{{\rm ps}}(t) \approx 1 - k(1-x)\left(1 - k(1-x)\frac{\chi}{\xi}\right)\exp{-\xi t}.
    \end{split}
\end{align}
Finally, making the substitutions $x \to \lambda$, $\chi \to \mathcal{N}$, and $\xi \to \mathcal{N}k$ for the $p,\lambda$-family, and $x \to -q/2$, $\chi \to \mathcal{N}(1+q/2)$, and $\xi \to \mathcal{N}k(1+q/2)$ for the $p,q$-family, we arrive at the following respective expressions for their second-order Glauber coherence functions
\begin{subequations}\label{g2_both_final}
    \begin{align}\label{g2lam_final}
        g^{(2)}_{{\rm ps},\lambda}(t) \approx 1 + \lambda(\lambda - 1)k\exp{-\mathcal{N}k|t|},
    \end{align}
    \begin{align}\label{g2q_final}
        g^{(2)}_{{\rm ps},q}(t) \approx 1 + \frac{q(1+q/2)}{2}k\exp{-\mathcal{N}k(1+q/2)|t|}.
    \end{align}
\end{subequations}

From this analysis, we can see that at zero time delay ($t=0$) the Glauber coherence functions are minimized with $\lambda = 0.5$ for the $p,\lambda$-family, and $q=-1.0$ for the $p,q$-family. In both of these instances, we find $g^{(2)}_{\rm ps}(0) =1- k/4$, while the decay rate of the photon correlations will be twice as large for the $p,\lambda$-family compared to the $p,q$-family (assuming a fixed $\mathcal{N}$). 

In the linearized regime, we require that $\partial p/\partial \mu = 0$ in order to maintain Heisenberg-limited coherence for both families of laser models (details on this point are provided in Sec~\ref{phase}, below). This requirement implies that $g^{(2)}_{\rm ps}(0)-1=\Theta(1/\mu^2)$ (the $\Theta$-notation here means that $g^{(2)}_{\rm ps}(0)-1$ scales as $1/\mu^2$), and the decay rate of the photon correlations is $\Theta(\mathcal{N}/\mu^2)$. The scaling of these quantities can be modified by relinquishing the requirement for Heisenberg-limited coherence but remaining within the regime for which Eqs.~(\ref{g2lam_final}) and (\ref{g2q_final}) are valid. That is, we may set $p=\Theta(\mu^{2-\epsilon})$ where $\epsilon\in(0,2]$. In this case, $g^{(2)}_{\rm ps}(0)$ would be below unity by a term that is $\Theta(1/\mu^\epsilon)$, while the decay rate of the photon correlations would be $\Theta(\mathcal{N}/\mu^{\epsilon})$.

\subsection{Spectrum of Intensity Correlations and Beam Sub-Poissonianity}\label{sec:specs_and_sub_p}

\begin{figure}[ht]
\includegraphics[width=1.0\columnwidth]{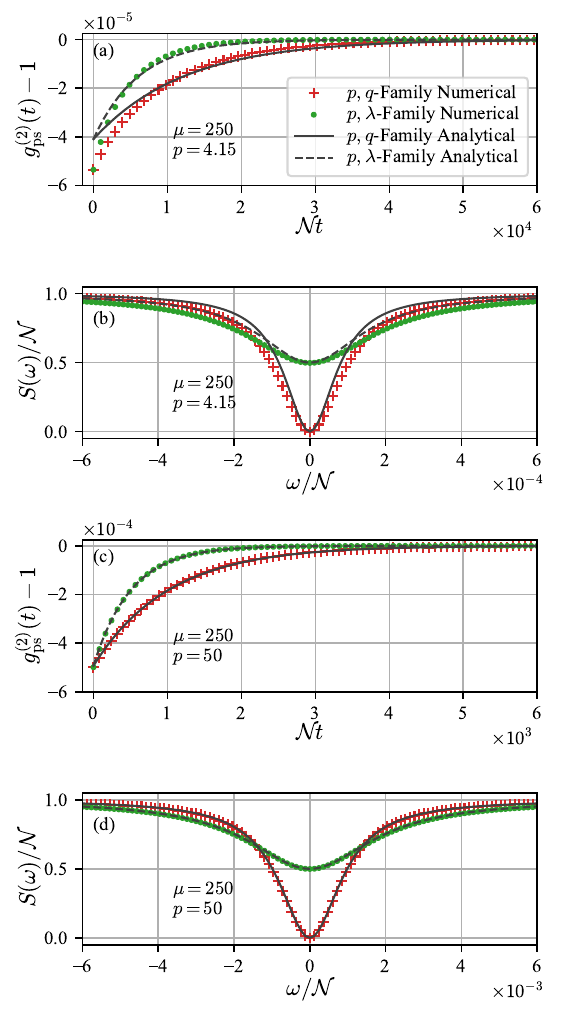}
\caption{\label{fig:G2_and_Spec} (a): Numerical and analytical calculations of the second-order correlation functions~(\ref{g2ps}) for the $p,\lambda$-family and $p,q$-family of laser models. Here, we consider the maximally sub-Poissonian case (i.e., $\lambda=0.5$ and $q=-1$, respectively for the two families), and have set $\mu = 250$ and $p=4.15$. (b) Corresponding intensity noise spectra computed from the data in (a). Panels (c) and (d) are the same as that shown in (a) and (b), respectively, but using $p=50$, which is clearly in the linearized regime.}
\end{figure}

With expressions for $g_{\rm ps}^{(2)}(t)$, we can readily evaluate Eqs.~(\ref{Q_t})~and~(\ref{spec intense}) to obtain the Mandel-$Q$ parameters and the intensity noise spectra for the two families of laser models. We find
\begin{subequations}
    \begin{align}\label{qlam_final}
        Q_{t}^{(\lambda)} = 2\lambda(\lambda - 1)\left[1+\frac{1}{\mathcal{N}k t}\left(\exp{-\mathcal{N}kt} - 1\right)\right],
    \end{align}
    \begin{align}\label{SIlam_final}
        S_{I}^{(\lambda)}(\omega) = \mathcal{N}\left(1+\frac{2k^2\lambda(\lambda-1)}{k^2+\left(\omega/\mathcal{N}\right)^2}\right),
    \end{align}
\end{subequations}
for the $p,\lambda$-family and
\begin{subequations}
    \begin{align}\label{qq_final}
        Q_{t}^{(q)} = q\left[1+\frac{\left(\exp{-\mathcal{N}k(1+q/2)t} - 1\right)}{\mathcal{N}k(1+q/2)t}\right],
    \end{align}
    \begin{align}\label{SIq_final}
        S_{I}^{(q)}(\omega) = \mathcal{N}\left(1+\frac{qk^2(1+q/2)^2}{[k(1+q/2)]^2+\left(\omega/\mathcal{N}\right)^2}\right),
    \end{align}
\end{subequations}
for the $p,q$-family.

In Fig~\ref{fig:G2_and_Spec}, we verify the validity of our analytical results, where we compare Eqs.~(\ref{g2lam_final})~and~(\ref{g2q_final}) (the second-order Glauber coherence functions), and Eqs.~(\ref{SIlam_final})~and~(\ref{SIq_final}) (the intensity noise spectra) with their exact counterparts, obtained from numerical calculations of the master equations~(\ref{lambda_master})~and~(\ref{q_master}). Here, we have chosen parameter values that yield the highest degree of sub-Poissonianity in both families of laser models and are comparing two scenarios; one for which $p=4.15$ and one for which $p=50$. 

For relatively low values of $p$, we can see clear deviations between the numerical and analytical results, but, for higher values of $p$, in the linearized regime, we see essentially perfect agreement. Note, however, that the numerical and analytical spectra match at $\omega=0$, regardless of the choice of $p$. This implies that the long-time ($t\to\infty$) Mandel-$Q$ parameter that is presented in Eqs.~(\ref{qlam_final})~and~(\ref{qq_final}) are valid outside of the linearized regime for which they have been derived. Indeed, the formulas $Q_{t\to\infty}^{(\lambda)} = 2\lambda(\lambda-1)$ and $Q_{t\to\infty}^{(q)} = q$ perfectly match with the numerical evaluations of this quantity that were carried out in Refs.~\cite{Ostrowski2022a,Ostrowski2022b}. Another notable aspect of the spectra shown in Fig~\ref{fig:G2_and_Spec} is that, while the depth of the spectrum is two times as large for the $p,q$-family compared to the $p,\lambda$-family, the spectral width of the noise reduction [specifically, the full width at half minimum of the inverted Lorentzian functions in Eqs.~(\ref{SIlam_final})~and~(\ref{SIq_final})] is smaller for the $p,q$-family by a factor of two.

\section{Phase Dynamics and Beam Coherence}
\label{phase}

In this section, we develop a more complete picture of the laser dynamics in the linearized regime. We show that the internal cavity dynamics for our two families approximates pure phase diffusion, in analogy to what is seen in standard laser systems~\cite{Carmichael_2010,Noh_2008}. However, unlike the coherent state in a standard laser, the pure cavity state undergoing phase diffusion in the Heisenberg-limited laser families has a much narrower phase distribution, and hence a much broader number distribution (i.e., it is highly phase squeezed). Note that while this internal state is phase squeezed, this is not the case for the output field because of the specific choice of coupling between between the laser cavity and its beam. Indeed, as was shown from our analysis in Section~\ref{sec:specs_and_sub_p}, the beam can instead be number squeezed.

Following the analysis of the internal cavity dynamics for our two families of laser models, we compute their first-order Glauber coherence functions, which yield a number of insights. Firstly, it allows for an analytical formula for the coherence of the beams to be obtained, where a direct connection can be drawn between this and the best-matching phase diffusion rate from the internal dynamics of the cavities. Secondly, while the long-time behaviour of $G^{(1)}(s,t)$ closely follows an exponential decay (as is the case for an ideal laser beam), we find small deviations from this at short times. These deviations are due to the departure from Poissonian photon number statistics in the beam, and can accordingly be expressed in terms of $g^{(2)}_{\rm ps}(t)$.

\subsection{Physically Realizable Ensembles and Pure Phase Diffusion}

For our analysis, we utilize the concept of PR ensembles~\cite{Wiseman1998}, framed within the context of pure phase diffusion. A PR ensemble refers to a stationary pure state ensemble of an open quantum system characterized by a distinct property: If the system is prepared in a pure state drawn from that ensemble, and the environment state is also initially pure, then it is possible to determine the system to be in a pure state of that ensemble \textit{at any later time} by monitoring the system's environment. In the linearized regime, we will show that our families of laser models can be accurately characterized by a set of pure states $\{\ket{\psi^\phi}\}$, with uniform measure over $\phi$, which constitute a PR ensemble. Furthermore, we demonstrate that the dynamics of $\phi$ in this regime correspond to pure diffusion. Before presenting these results, we provide a brief review of the relevant theoretical framework.

Consider an open quantum system, where the dynamics of which can be written in the form of a Markovian master equation $\dot\rho = \mathcal{L}\rho$, with $\mathcal{L}$ representing the Liouvillian superoperator. Let us also suppose that a unique stationary mixed state $\rho_{\rm ss} = \rho(t\to\infty)$ exists for this system, with $\mathcal{L}\rho_{\rm ss} = 0$. At any given time, this state of the system can be expressed as a convex combination of pure states in infinitely many ways. Each decomposition consists of ordered pairs of rank-one projectors onto pure states $\Pi_n$ and probabilities $p_n$, such that $\rho_{\rm ss} = \sum^{N-1}_{n=0}p_n\Pi_n$. Only some decompositions, however, are PR. For these decompositions, it possible to devise a way to monitor the system's environment---without affecting the average evolution of the system---such that the system will only ever be in one of the states $\Pi_n$, and the proportion of time that the system will spend in that state is equal to $p_n$ in the long time limit. The pertinent result from Ref.~\cite{Wiseman2001} (see also Ref.~\cite{Karasik_2011}) is that a \textit{discrete ensemble} $\{(\Pi_n,p_n)\}_{n=0}^{N-1}$ representing $\rho$ is PR iff there exist rates $\kappa_{mn}\geq0$ such that
\begin{align}\label{PR_diff_condition}
    \forall\,n, \quad\quad \mathcal{L}\Pi_n = \sum_{m=0}^{N-1}\kappa_{mn}\left(\Pi_m - \Pi_n\right).
\end{align}

The ensemble of pure states $\left\{(\varrho^\phi,d\phi/2\pi)\right\}$ that we consider is instead defined with respect to the \textit{continuous} phase variable $\phi\in[0,2\pi)$ with uniform measure. Specifically, we make the ansatz 
\begin{align}\label{pure_cavity_state}
    \begin{split}
        & \varrho^\phi = |\psi^\phi\rangle\langle\psi^\phi|, \quad {\rm with} \quad
        \ket{\psi^\phi} := \sum_{n = 0}^{D-1}\sqrt{\rho_n}e^{i\phi n}\ket{n},
    \end{split}
\end{align}
and $\rho_n:=\bra{n}\rho_{\rm ss}\ket{n}$ defines the coefficients the steady state, where $\mathcal{L}\rho_{\rm ss}=0$. It should be clear that that $\varrho^\phi$ satisfies the constraint that the uniformly weighted ensemble reproduces the incoherent steady state; $\rho_{\rm ss} = \int_0^{2\pi}\varrho^\phi d\phi/2\pi$. Additionally, we wish to consider the situation where the dynamics of the ensemble are that of pure phase diffusion. This is where an initial state $\varrho^\phi$ would evolve in a time $t$ into a Gaussian mixture of states $\nu$ from the same family; i.e., $\nu = \int_{\mathbb{R}} d\varphi\,g(\varphi;\phi,\ell t)\varrho^{\varphi\,{\rm mod}\,2\pi}$, where $g(\varphi;\phi,\ell t)$ denotes a Gaussian in $\varphi$ with mean $\phi$ and variance $\ell t$ ($\ell$ representing the diffusion rate). Under these conditions, states with the general form $\varrho^{\phi} = \sum_{m,n}e^{i\phi(m-n)}|\rho_{mn}|\ketbra{m}{n}$ [which includes the ansatz~(\ref{pure_cavity_state})] will evolve according to the diffusion equation
\begin{align}\label{PR_diff_phase}
    \forall\,\phi \quad\quad \mathcal{L}\varrho^\phi = \left.\frac{\ell}{2}\frac{d^2\varrho^\varphi}{d\varphi^2}\right|_{\varphi = \phi}.
\end{align}

We will show that Eq.~(\ref{PR_diff_phase}) represents an analogous condition to Eq.~(\ref{PR_diff_condition}) for the continuous ensemble $\left\{(\varrho^\phi,d\phi/2\pi)\right\}$ to be PR. This can be done by assigning $\Pi_n=\varrho^{\phi_{n}}$, where $\phi_n = 2\pi n/N$, and choosing the rates $\kappa_{mn}$ according to
\begin{align}
    \begin{split}
        & \kappa_{mn} = \\ &
\begin{cases}
    \frac{\ell}{2}\left(\frac{N}{2\pi}\right)^2(\delta_{m-1,0} + \delta_{m-(N-1),0}) & \hspace{0.7cm} \text{if } n=0, \\
    \frac{\ell}{2}\left(\frac{N}{2\pi}\right)^2(\delta_{m-1,n} + \delta_{m+1,n}) & \hspace{-0.68cm} \text{if } 0< n < N-1, \\
    \frac{\ell}{2}\left(\frac{N}{2\pi}\right)^2(\delta_{m+(N-1),N-1} + \delta_{m+1,N-1}) & \text{if } n = N-1.
\end{cases}
    \end{split}
\end{align}
Substituting these values for $\kappa_{mn}$ into Eq.~(\ref{PR_diff_condition}) gives
\begin{align}\label{mid_step}
    \begin{split}
        \forall\,n, \quad \mathcal{L}\varrho^{\phi_n} = \frac{\ell}{2}\left(\frac{N}{2\pi}\right)^2\bigg( & \varrho^{(\phi_n-2\pi/N)\,{\rm mod}\,2\pi} - 2\varrho^{\phi_n} \\ & + \varrho^{(\phi_n+2\pi/N)\,{\rm mod}\,2\pi}\bigg),
    \end{split}
\end{align}
and then taking the limit $N\to\infty$ yields Eq.~(\ref{PR_diff_phase}). This reveals that Eq.~(\ref{PR_diff_phase}) is a continuous counterpart to a version of the discrete equation~(\ref{PR_diff_condition}). Moreover, given that it is also of the form of a diffusion equation, it can be considered to represent the condition for the ensemble $\left\{(\varrho^\phi,d\phi/2\pi)\right\}$, with the specific ansatz~(\ref{pure_cavity_state}), to be PR \textit{and} describable with dynamics of pure phase diffusion of $\phi$.

\subsection{Matching Cavity Dynamics with Pure Phase Diffusion}\label{match_cav_diff}

Our goal is to show that the dynamics given by the $p,\lambda$- and $p,q$-families of laser models satisfy (to an arbitrarily good approximation) Eq.~(\ref{PR_diff_phase}) within the linearized regime. If this can be done, then we can confidently assert that the ensemble $\left\{(\varrho^\phi,d\phi/2\pi)\right\}$ is PR, and that the dynamics of $\varrho^\phi$ is given by pure phase diffusion within that regime. To do this, we first note that the RHS of Eq.~(\ref{PR_diff_phase}) can be expressed as
\begin{align}\label{pure_diff_deravitive_finite}
    \left.\frac{d^2\varrho^\varphi}{d\varphi^2}\right|_{\varphi=\phi} = - \sum_{m,n=0}^{D-1}(m-n)^2e^{i\phi(m-n)}\sqrt{\rho_m\rho_n}\ketbra{m}{n}.
\end{align}
Then looking at the LHS of~(\ref{PR_diff_phase}), we make the respective substitutions $\mathcal{L}\to\mathcal{L}^{(p,\lambda)}$ and $\mathcal{L}\to\mathcal{L}^{(p,q)}$ for the two families of laser models. By focusing on the case where $1\ll p\ll\mu^2$, we derive the following expressions to second order in $k$,
\begin{subequations}
    \begin{align}\label{diff_lam}
        \begin{split}
            \mathcal{L}^{(p,\lambda)}\varrho^{\phi} \approx & -\frac{1}{2} \left\{(2\lambda^2-2\lambda+1)\frac{\mathcal{N}\pi^4p^2}{64\mu^4}\right\} \\ & \times\sum_{m,n}(m-n)^2 e^{i\phi(m-n)}\sqrt{\rho_m\rho_n}\ketbra{m}{n},
        \end{split}
    \end{align}
    \begin{align}\label{diff_q}
        \begin{split}
            \mathcal{L}^{(p,q)}\varrho^{\phi} \approx & -\frac{1}{2} \left\{(1+q/2)^2\frac{\mathcal{N}\pi^4p^2}{64\mu^4}\right\} \\ & \times\sum_{m,n}(m-n)^2 e^{i\phi(m-n)}\sqrt{\rho_m\rho_n}\ketbra{m}{n}.
        \end{split}
    \end{align}
\end{subequations}
Details of these calculations are provided in Appendix~\ref{apendB}. 

Eqs.~(\ref{diff_lam})~and~(\ref{diff_q}) are of the same form as Eq~(\ref{pure_diff_deravitive_finite}). This implies that the dynamics of the two families approach pure phase diffusion for $1\ll p\ll\mu^2$. The diffusion rates are given by the expressions within the curly braces of~(\ref{diff_lam})~and~(\ref{diff_q}),
\begin{subequations}\label{ell_both}
    \begin{align}
        \ell^{(p,\lambda)} := (2\lambda^2-2\lambda + 1)\frac{\mathcal{N}\pi^4p^2}{64\mu^4},
    \end{align}
    \begin{align}
        \ell^{(p,q)} := (1+q/2)^2\frac{\mathcal{N}\pi^4p^2}{64\mu^4}.
    \end{align}
\end{subequations}

In Fig.~\ref{fig:diff_verify}, we show quantitative results regarding the validity of these phase diffusion approximations by demonstrating the extent to which Eq.~(\ref{PR_diff_phase}) holds. There, we plot the relative distance between the matrices $\mathcal{L}\varrho^{\phi=0}$ and $\left.(\ell/2)d^2\varrho^\varphi/d\varphi^2\right|_{\varphi=0}$ for our families of laser models, verifying that there exists an $\ell$ such that these two matrices closely match in the linearized regime. Explicitly, we give strong evidence for the claim that
\begin{align}\label{to_verify}
    \begin{split}
        \lim_{\mu\to\infty}R\left(\left.\frac{\ell}{2}\frac{d^2\varrho^\varphi}{d\varphi^2}\right|_{\varphi=\phi},\mathcal{L}\varrho^\phi;\mathcal{L}\varrho^\phi\right) = O(p^{-1}),
    \end{split}
\end{align}
where $R(\hat{c},\hat{d};\hat{e}) := ||(\hat{c}-\hat{d})||/||\hat{e}||$ for arbitrary operators $\hat{c}$, $\hat{d}$ and $\hat{e}$, and $||\hat{f}||=[{\rm Tr}(\hat{f}\hat{f}^\dagger)]^{1/2}$ is the Hilbert-Schmidt norm for the operator $\hat{f}$~\cite{bengtsson2006}, and $\ell$, the best-matching diffusion rate for the respective family of models, is given in Eqs.~(\ref{ell_both}).

Considering the plot in Fig.~\ref{fig:diff_verify}c, the relative distance for the $p,q$-family does not appear to follow a power law for relatively large values of $p$, in contrast to the other presented cases (Figs.~\ref{fig:diff_verify}a~and~\ref{fig:diff_verify}b). Regardless, this data does appear to be approaching the power law given by the black line in the asymptotic limit $\mu\to\infty$. We verify this in the inset of that plot, which depicts, on a log-log scale, the fitted trendline of $R_0=5.90p^{-1.17}$ subtracted from the relative distance $R$, for fixed $p=50$. This data indicates that the relative distance converges toward the fitted black trendline in the larger plot of Fig.~\ref{fig:diff_verify}c, as $\mu\to\infty$ and $p=O(1)$ in $\mu$.

\begin{figure}[H]
\includegraphics[width=0.85\columnwidth]{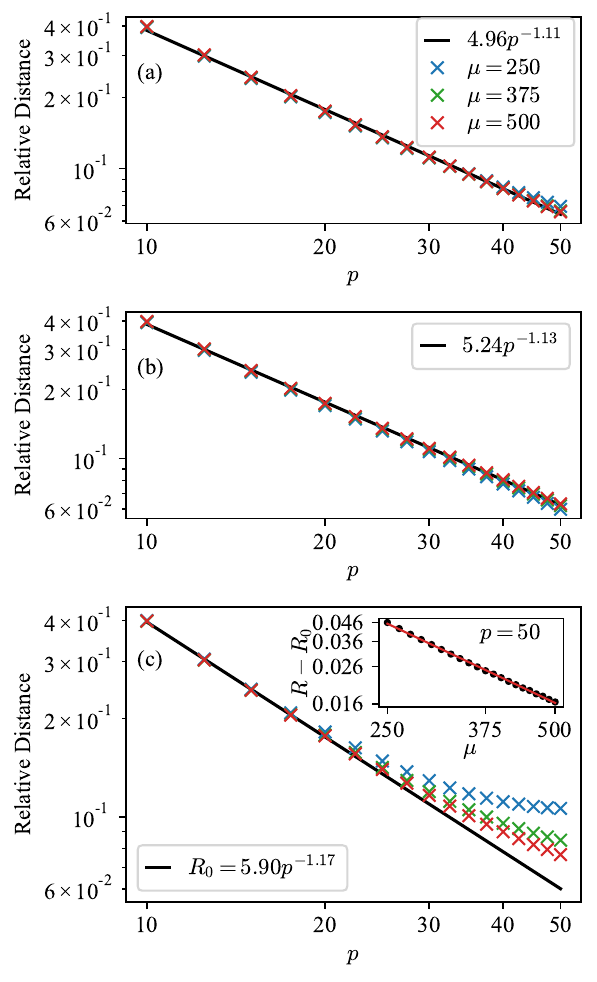}
\caption{\label{fig:diff_verify} Relative distance, as defined in the LHS of Eq.~(\ref{to_verify}) against $p$. (a) considers the Poissonian limit for the $p,\lambda$-family~(\ref{lambda_master}) and $p,q$-family~(\ref{q_master}), where respectively one sets $\lambda\to0$ and $q\to0$ (in this scenario, the two families are equivalent). (b) considers the the maximally sub-Poissonian case for the $p,\lambda$-family~(\ref{lambda_master}), where $\lambda\to0.5$. (c) considers the the maximally sub-Poissonian case for the $p,q$-family~(\ref{q_master}), where $q\to-1$. Crosses present numerical data, while the solid black lines are power-law fits to the data pertaining to $\mu = 500$ and $p\in[10,50]$ in panels (a) and (b), and $\mu = 500$ and $p\in[10,20]$ in panel (c). The inset in (c) considers a fixed value of $p=50$ and plots the fitted power law $R_0 = 5.90p^{-1.17}$ [black line in (c)] subtracted from the relative distance as a function of $\mu$. Black dots depict numerical data and the red line is the fitted power law $159\mu^{-1.48}$.}
\end{figure}

\subsection{Phase Squeezing in the Laser Cavity State}

The phase diffusion rates~(\ref{ell_both}) are vastly smaller, in terms of their scaling with $\mu$ compared to the ideal phase diffusion rate of a standard laser~\cite{Wiseman_1999}. This is closely linked to the fact that there is substantial phase squeezing in the pure cavity state~(\ref{pure_cavity_state}) undergoing phase diffusion. Here, we show this explicitly by computing the quantity $1-|\exp{i\hat{\varphi}_{\rm PB}}|^2$ for the state~(\ref{pure_cavity_state}), using the approximate form of $\rho_{\rm ss}$ given in Eq.~(\ref{ss_lin}). The operator $\hat{\varphi}_{\rm PB}$ is the Pegg-Barnett phase operator~\cite{Pegg_1988}, and the quantity $1-|\exp{i\hat{\varphi}_{\rm PB}}|^2$ represents a convenient measure of phase spread, being periodic over $2\pi$. In the limit of small phase uncertainty (which is the relevant case for us), one also has $1-|\langle\exp{i\hat{\varphi}_{\rm PB}}\rangle|^2\sim\langle(\Delta \hat{\varphi}_{\rm PB})^2\rangle$. Note that $\hat{\varphi}_{\rm PB}$ requires the choice of a reference phase $\theta_0$, which we choose as $\theta_0=-\pi$. This ensures the uncertainty relation $\langle(\Delta \hat{n})^2\rangle\langle(\Delta \hat{\varphi}_{\rm PB})^2\rangle\geq1/4$ for a phase probability distribution that is well localized about $\phi=0$. We thus also choose $\phi=0$ in the state~(\ref{pure_cavity_state}), for which we may write~\cite{Bandilla_1991}
\begin{align}\label{phase_variance_pure_state_1}
    \underset{D\to\infty}{\lim}\left(1-|\exp{i\hat{\varphi}_{\rm PB}}|^2\right) = 1-\left(\sum_{n=0}^{D-2}\sqrt{\rho_{n+1}\rho_{n}}\right)^2.
\end{align}
Substituting $\rho_n$ with the elements of the approximate Gaussian steady state~(\ref{ss_lin}), we find
\begin{align}\label{phase_variance_pure_state_2}
    1-|\exp{i\hat{\varphi}_{\rm PB}}|^2 = \frac{k}{4}+O(k^2).
\end{align}

In the linearized regime, the state~(\ref{pure_cavity_state}) therefore has the phase variance $\langle(\Delta\hat{\varphi}_{\rm PB})^2\rangle\sim k/4=\pi^2p/16\mu^2$ for $k\ll1$. If $p=O(1)$ in $\mu$, this phase variance is quadratically smaller than that of a coherent state, where $\langle(\Delta\hat{\varphi}_{\rm PB})^2\rangle\sim1/4\mu$ for $\mu\gg1$. Similar to the results presented in Ref.~\cite{Liu2021}, we hence see substantial phase squeezing in a PR cavity state for a laser that surpasses the standard quantum limit for coherence. It is also notable that that the number variance of the pure state~(\ref{pure_cavity_state}) is $\langle(\Delta \hat{n})^2\rangle \approx 1/k$ in the linearized regime, meaning it is also a minimum uncertainty state, as $\langle(\Delta \hat{n})^2\rangle\langle(\Delta \hat{\varphi}_{\rm PB})^2\rangle\sim1/4$ for $k\ll1$.

The key feature of the laser models that enables the behavior described above lies in the specific form of the operators governing their interactions with the environment. These operators are designed to preserve the phase statistics within the device more effectively than the creation ($\hat{a}^\dagger$) and annihilation ($\hat{a}$) operators that characterize gain and loss in a standard laser. Although these engineered operators are designed to raise and lower the cavity number by one, their coefficients in the number basis can be viewed as being much \textit{flatter} than those of $\hat{a}^\dagger$ and $\hat{a}$.

This point can be highlighted by using the $p,\lambda$-family of laser models within the linearized regime as an example. From the master equation (\ref{lam_mast_lin}), the pure cavity state defined in Eq.~(\ref{pure_cavity_state}) can be shown to evolve as
\begin{align}
    \begin{split}\label{lam_dynamics_pure_cav}
        \frac{d\varrho^\phi}{dt} = & - \frac{\mathcal{N}}{2}\sum_{n,m=1}^{D-2}\bigg{\{}\left[{G}_{n+1}(p,\lambda) - {G}_{m+1}(p,\lambda)\right]^2 \\ & \hspace{2.1cm} + \left[{L}_{n}(p,\lambda) - {L}_{m}(p,\lambda)\right]^2\bigg{\}} \\ & \hspace{1.75cm} \times\sqrt{\rho_n\rho_m}e^{i\phi(n-m)}\ket{n}\bra{m}.
    \end{split}
\end{align}
In Section~\ref{match_cav_diff}, it was verified that Eq.~(\ref{lam_dynamics_pure_cav}) is equivalent to a diffusion equation for the variable $\phi$ (see also Appendix~\ref{apendB}). The diffusion rate $\ell^{(p,\lambda)}$ can be obtained from the term in the curly braces of Eq.~(\ref{lam_dynamics_pure_cav}) by setting $m=n-1$ and taking the limit $k\to0$. Treating $n$ as a continuous variable, we find
\begin{align}\label{ell_intuitive}
    \begin{split}
        \frac{\ell^{(p,\lambda)}}{\mathcal{N}} = \underset{k\to0}{\rm lim}\left\{ \left[\frac{dG(n,p,\lambda)}{dn}\right]^2+ \left[\frac{dL(n,p,\lambda)}{dn}\right]^2\right\}.
    \end{split}
\end{align}
Writing the diffusion rate in the manner given in Eq.~(\ref{ell_intuitive}) highlights its close relation to the flatness of both the gain and loss coefficients. Intuitively, this is reasonable because the closer these coefficients are to being constant in the number basis of the cavity, the closer they will be to commuting with the phase operator $\hat{\varphi}_{\rm PB}$, and their action on a particular cavity state will therefore better preserve its phase statistics. Overall, the effect of this will be a reduction in the rate of phase diffusion, and hence, an increase in $\mathfrak{C}$.

\subsection{The First-Order Glauber Coherence Function and Laser Coherence Calculations}

Above, we found that the two families of laser models evolve according to pure phase diffusion in the linearized regime. The best matching diffusion rates that were obtained from this analysis are given in Eqs.~(\ref{ell_both}). In this subsection, we form a connection between these diffusion rates and the laser coherence $\mathfrak{C}$ for these families, by evaluating their first-order Glauber coherence functions.

First, we use similar methods from Refs.~\cite{Davidovich_1987,Bergou_1989b} to derive a Fokker-Planck-like equation for evolution of the upper diagonal elements of the density matrix $\rho_{n-1,n}:=\bra{n-1}\rho\ket{n}$. In Appendix~\ref{off-diag Append}, we show that by taking the continuum limit, with $\rho_{n-1,n}\to p_{\rm coh}(n,t)$, $G_{n}(p,\lambda)\to G(n,p,\lambda)$ and $L_{n}(p,\lambda)\to L(n,p,\lambda)$, the equations of motion of $p_{\rm coh}(n,t)$ for the $p,\lambda$- and $p,q$-families of laser models in the linearized regime be approximated with
\begin{align}\label{general OE}
    \begin{split}
        \dot p_{\rm coh}(n,t) = & - \frac{\ell}{2}p_{\rm coh}(n,t) \\ & + \frac{\partial}{\partial n}\left[\xi(n-\tilde\mu)p_{\rm coh}(n,t)\right] + \chi\frac{\partial^2}{\partial n^2}p_{\rm coh}(n,t).
    \end{split}
\end{align}
Here, $\tilde\mu = \mu+1/2$, and to specialize to the $p,\lambda$-family, one would make the substitutions $\ell\to\ell^{(p,\lambda)}$, $\xi\to \mathcal{N}k$ and $\chi\to\mathcal{N}$, while for the $p,q$-family, one would set $\ell\to\ell^{(p,q)}$, $\xi\to \mathcal{N}k(1+q/2)$ and $\chi\to\mathcal{N}(1+q/2)$. Eq.~(\ref{general OE}) resembles an Ornstein-Uhlenbeck process, but with the addition of a decay term $-\ell p
_{\rm coh}(n,t)/2$ in the dynamics. For an initial condition $p_{\rm coh}(m,t|n,0)=\delta(m-n)$, we find a solution to Eq.~(\ref{general OE}), for $t>0$
\begin{align}\label{solution_coh_cts}
    \begin{split}
        p_{\rm coh}(m,t|n,0) = & \frac{e^{-\ell t/2}}{\sqrt{2\pi\sigma(t)^2}} \\ & \times\exp{\frac{-\left[m-\tilde\mu - (n-\tilde\mu)e^{-\xi t}\right]^2}{2\sigma(t)^2}}.
    \end{split}
\end{align}
Here, $\sigma(t)^2$ is the same as that given in Eq.~(\ref{solution_cts}).

Now that we have obtained a solution for the evolution of the coherences within the system, we are in a position to evaluate the first-order Glauber coherence function. In a similar vein to what was done for the second-order Glauber coherence function in Section~\ref{g2_calculation}, we write down a general equation for $G^{(1)}(s+t,s)$ in the continuum limit (assuming the system resides in its steady state at time $s$)
\begin{align}\label{g1_continuum}
    \begin{split}
        G^{(1)}(s+t,s) \approx \mathcal{N}\int_0^\infty\int_0^\infty dm\,dn & L(n,k,x)L(m, k,x) \\ & \times p(n)p_{\rm coh}(m,t|n,0),
    \end{split}
\end{align}
where $L(n,k,x)$ is the continuum form of Eq.~(\ref{loss_coeff_linear_general}) and $p(m) = p(m,t\to\infty|n,0)$ is the same as that in~(\ref{g2_generalcts}). It is the case for both families of laser models that $\chi/\xi=O(1/k)$, which is small for $1\ll p\ll\mu^2$. Thus, following an expansion of the loss functions to second order in $k$, $L(n,k,x) \approx 1 + (k(1-x)/2)(n-\mu) - (k^2(1-x)^2/8)(n-\mu)^2$, Eq.~(\ref{g1_continuum}) can be evaluated, yielding
\begin{align}\label{g1_general_1}
    \begin{split}
        G^{(1)}(s+t,s) = \bigg[ &1 + \frac{k(1-x)}{4}\left(1 - k(1-x)\frac{\chi}{\xi}\right)\\ & \times\left(1-e^{-\xi |t|}\right) + O(k^2) \bigg]\mathcal{N}e^{-\ell |t|/2}.
    \end{split}
\end{align}

It is illuminating to employ Eq.~(\ref{g2_generalcts2}), which allows us to express the general form of $G^{(1)}(s,t)$ as
\begin{align}\label{g1_general_2}
    G^{(1)}(s+t,s) \approx \mathcal{N}e^{-\ell |t|/2}\bigg[ 1 + \frac{1}{4}\left(g_{\rm ps}^{(2)}(t) - g_{\rm ps}^{(2)}(0)\right)\bigg],
\end{align}
where we have neglected the $O(k^2)$ terms in the square braces of~(\ref{g1_general_1}). Here, we can see deviations within the short-time behavior of $G^{(1)}(s+t,s)$ compared to what what one would expect by considering an ideal laser beam [given by Eq.~(\ref{ideal laser})], which would instead be the exponential decay $G^{(1)}(s+t,s) = \mathcal{N}e^{-\ell |t|/2}$. These deviations occur as the two families of laser models depart from exhibiting Poissonian beam photon statistics, and are significant within the photon correlation time [the timescale associated with the decay in $g_{\rm ps}^{(2)}(t)$], which is much shorter than the coherence time of the laser (the reciprocal of the phase diffusion rate of the laser cavity). Thus, on longer timescales, Eq.~(\ref{g1_general_1}) is indistinguishable from what one would obtain from an ideal laser beam. In Fig.~\ref{fig:G1}, we compare our analytically-derived formulas for $G^{(1)}(s+t,s)$ with numerical results, from which excellent agreement is seen in the appropriate parameter regime.

\begin{figure}[H]
\includegraphics[width=1.0\columnwidth]{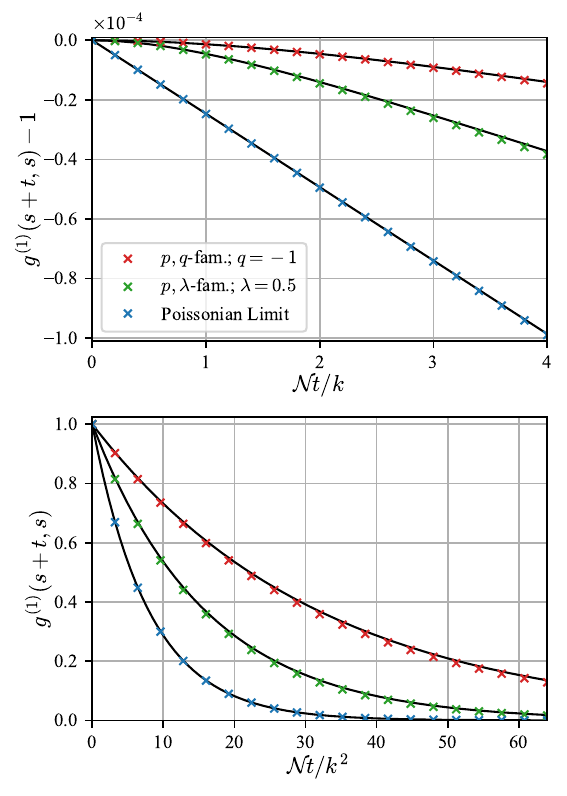}
\caption{\label{fig:G1}Top: Short-time behaviour of the First-order Glauber coherence function for the $p,\lambda$- and $p,q$-families of laser models for $p=50$ and $\mu = 250$. Crosses represent numerical data and black curves depict the analytical formula~(\ref{g1_general_2}). The ``Poissonian limit" refers to the case in which the two families are described by the same master equation (respectively, with $\lambda\to0$ and $q\to0$). Bottom: Same as the top plot, but considering a much larger timescale, on the same order of magnitude as the coherence time of the laser cavities.}
\end{figure}

Because of the vast difference between the photon correlation time and the coherence time, the coherence $\mathfrak{C}$ can still be expressed in an intuitive manner. To see this explicitly, we evaluate Eq.~(\ref{coh g1}) using Eq.~(\ref{g1_general_1}) and find
\begin{align}\label{coh_intuitive_linearized}
    \mathfrak{C} = \frac{4\mathcal{N}}{\ell} + O(1/k).
\end{align}
Recall from Eqs.~(\ref{ell_both}), that $\ell=\Theta(k^2)$ for our families of laser models; hence, in the linearized regime, the coherence can be written as being proportional to flux divided by the phase diffusion rate of the laser cavity. This means that, within this regime, we may still interpret the coherence in the same manner as we did for the ideal laser beam state discussed in Section~\ref{sec:coherence_general}, even though there can be a significant degree of number squeezing in the beam. That is, by taking $1/\ell$ as the coherence time of the laser, the coherence is thus roughly equal to the number of photons emitted into the beam by the laser that are mutually coherent. 

Using $\ell^{(p,\lambda)}$ and $\ell^{(p,q)}$ given in Eqs.~(\ref{ell_both}), we can obtain explicit formulas for the coherence of the $p,\lambda$- and $p,q$-families of laser models in the linearized regime,
\begin{subequations}
    \begin{align}\label{coh_lam_lin}
        \mathfrak{C}^{(p,\lambda)}_{\rm lin} = \frac{1}{(2\lambda^2-2\lambda+1)}\frac{256\mu^4}{\pi^4p^2},
    \end{align}
    \begin{align}\label{coh_q_lin}
        \mathfrak{C}^{(p,q)}_{\rm lin} = \frac{1}{(1+q/2)^2}\frac{256\mu^4}{\pi^4p^2}.
    \end{align}
\end{subequations}
It is worth comparing these formulas with the heuristically-derived expressions for the coherence of the $p,\lambda$- $p,q$-families in Eqs.~(58a) and (58b) of Ref.~\cite{Ostrowski2022a}. Those heuristic formulas show very good agreement with numerical evaluations for $p>3$ and $\mu\to\infty$. In the limit $p\to\infty$, those heuristic formulas are asymptotically equivalent to Eqs.~(\ref{coh_lam_lin}) and (\ref{coh_q_lin}) given above. In Fig~\ref{fig:coh}, we show this agreement graphically, while also comparing with numerical calculations of $\mathfrak{C}$. Given a relatively large value for $\mu$, we can see that moderately-large values of $p$ yield very good agreement between these numerically-, heuristically- and analytically-derived quantities.

\section{A Tighter Upper Bound for Laser Coherence}
\label{tighter ub}

In this section we provide a final result: an upper bound for laser coherence that is tighter by a factor of $3/8$ than that proven in Ref.~\cite{Baker2020}. The proof of this new upper bound follows from the consideration of a stricter, yet comparatively more elegant condition on the beam. Namely, it requires that the beam be \textit{equivalent} to a coherent state undergoing pure phase diffusion [the so-called ideal laser beam state~(\ref{ideal laser})], and therefore constraints the Glauber coherence functions of the beam to all orders. In contrast, the derivation of the upper bound on $\mathfrak{C}$ in previous works~\cite{Baker2020,Ostrowski2022a,Ostrowski2022b} required constraints on only the first- and second-order Glauber coherence functions. 

This stricter beam requirement means that demonstrating achievability (i.e., finding a laser model that both saturates the Heisenberg-limited scaling of $\mathfrak{C}$ and satisfies the beam constraints required for the upper bound) is more difficult. However, we provide strong evidence that such a laser model indeed exists. This laser model can be realized from either of the $p,\lambda$- and $p,q$-laser families studied throughout this work, upon specializing to the Poissonian case, and the linearized regime.

\begin{figure}[H]
\includegraphics[width=1.0\columnwidth]{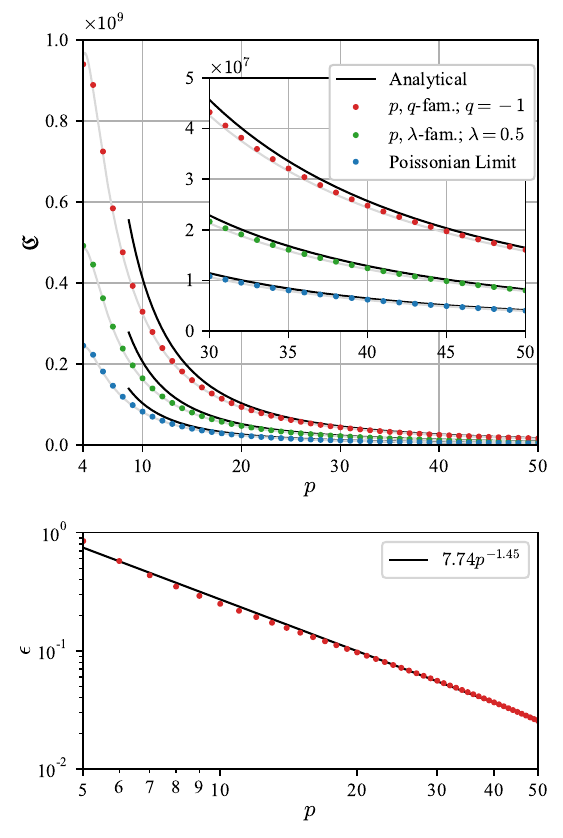}
\caption{\label{fig:coh}Top: Coherence of the $p,\lambda$- and $p,q$-families of laser models as a function of the parameter $p$, with $\mu = 250$. Numerical data is given by the coloured markers, while the analytically-derived formulas of Eqs.~(\ref{coh_lam_lin})~and~(\ref{coh_q_lin}) are represented by the black curves. Light-grey curves depict heuristically derived formulas for the coherence from Ref.~\cite{Ostrowski2022a} [see Eqs.~(58a) and (58b) in that work]. The ``Poissonian limit" refers to the case in which the two families are described by the same master equation (respectively, with $\lambda\to0$ and $q\to0$). Bottom: Relative error $\epsilon$ (red dots) between the numerical and analytical data given in the top plot. For a given value of $p$, the data points are essentially the same, regardless of the particular family considered. Black line shows a power law fitted to data pertaining to parameter values of $p\geq20$.}
\end{figure}

\subsection{Proof of the Upper Bound for $\mathfrak{C}$}

We proceed by explicitly stating the four conditions on a laser device producing a beam, for which we will derive the tighter upper bound on $\mathfrak{C}$; these are as follows.
\begin{enumerate}
     \item \textbf{One Dimensional Beam}---The beam propagates away from the device in one direction at a constant speed, occupying a single transverse mode and polarisation. The beam can therefore be described by a scalar quantum field with the annihilation operator $\hat{b}(t)$ satisfying $[\hat{b}(t),\hat{b}^\dagger(t')] = \delta(t-t')$.
    \item \textbf{Endogenous Phase}---Coherence in the beam proceeds only from coherence in the excitations within the device. Specifically, a phase shift imposed on the
    laser state at some time $T_0$ will lead, in the future, to the
    same phase shift on the beam emitted after time $T_0$, as well
    as on the laser state.
    \item \textbf{Stationary Beam Statistics}---The statistics of the device and beam have
    a long-time limit that is unique and invariant under time
    translation.  This means that the mean excitation number within the device, $\langle \hat{n}_{\rm c} \rangle$, has a unique stationary value $\mu$. 
    \item \textbf{Ideal Laser Beam---} The stationary, one dimensional beam produced by the device is \textit{statistically equivalent} to an ideal beam state. That is, the beam can be considered to be in an eigenstate $|\beta(t)\rangle$ of $\hat{b}(t)$, with an eigenvalue $\beta(t)=\sqrt{\mathcal{N}}e^{i\sqrt{\ell}W(t)}$, where $W(t)$ represents a Wiener process.
\end{enumerate}
The first three of these conditions are the same as those used in Refs.~\cite{Ostrowski2022a,Baker2020,Ostrowski2022b}; we therefore refer the reader to those references for details on each of those conditions. With our fourth condition, the term ``\textit{statistically equivalent}" means that the state of the beam produced by the laser device matches the ideal beam state $\ket{\beta(t)}$ to an arbitrarily good approximation. Details regarding the approximation that we appeal to will be given below in Sec.~\ref{p_fam_revisit}, upon the specification of a particular laser model. It is noteworthy here, however, that this condition is stricter than the fourth conditions used in the theorems of Refs.~\cite{Ostrowski2022a,Baker2020,Ostrowski2022b}, which instead place constraints only on the first- and second-order Glauber coherence functions of the beam. A beam that matches $\ket{\beta(t)}$ will guarantee that its first- and second-order Glauber coherence functions also match those of $\ket{\beta(t)}$, as required by those works.

With these four conditions, we now prove the following theorem.
\newtheorem{theorem}{Theorem}
\begin{theorem}[Upper bound on $\mathfrak{C}$ for an ideal laser beam]\label{thm_1}
For a device which satisfies conditions 1--4 stated above, the coherence is bounded from above by
\begin{align}\label{upper_bound}
    \mathfrak{C} \lesssim \frac{1}{4}\left|\frac{3}{z_A}\right|^6\mu^4,
\end{align}
in the asymptotic limit $\mu\to\infty$, where $\mu$ is the mean number of excitations stored within the device, and $z_A\approx-2.338$ is the first zero of the Airy function.
\end{theorem}

\begin{proof}

In a similar manner to the proofs presented in Refs.~\cite{Baker2020,Ostrowski2022a,Ostrowski2022b}, the upper bound that we derive follows from the comparison of two methods for estimating an encoded phase on the laser device. In particular, we first consider a filtering measurement to be performed on the beam, which encodes a phase onto the state of the device. Then, we compare the precision of two measurement schemes used to estimate that encoded phase; one being a retrofiltering measurement on the beam and the other being a direct measurement on the device. Due to the fact that the precision of the retrofiltering measurement cannot outperform the direct measurement to estimate the encoded phase, the upper bound on $\mathfrak{C}$ will follow from known bounds on phase estimation. Unlike those previous proofs that utilize heterodyne measurements to perform the filtering and retrofiltering procedures, we will consider \textit{optimal} filtering and retrofiltering measurements to carry out these tasks~\cite{Laverick_2018}. It is this distinction that leads to a tighter upper bound on $\mathfrak{C}$.

We proceed by invoking Condition 3 (Stationary Beam Statistics) and suppose that, at some time $T$, that the laser device resides in its unique steady state $\rho_{\rm c}(T) = \rho_{\rm c}^{\rm ss}$. Given Condition 2 (Endogenous Phase) and Lemma 2 of Ref.~\cite{Baker2020}, we know that this steady state is invariant under all optical phase shifts. That is, for all $\theta$, $\mathcal{U}^\theta_{\rm c}(\rho_{\rm c}^{\rm ss}) = \rho_{\rm c}^{\rm ss}$, where $\mathcal{U}^\theta_{\rm c}(\bullet):=e^{i\theta\hat{n}_{\rm c}}\bullet e^{-i\theta\hat{n}_{\rm c}}$, and $\hat{n}_{\rm c}$ is the generator of phase shifts on the Hilbert space of the device.

Next, we invoke Condition 4 (Ideal Laser Beam) and suppose that a filtering measurement is performed on this coherent beam over the time interval $[T,T')$, yielding an estimate $\phi_F$ of the true beam phase $\phi(T')$. We consider an optimal filtering measurement, which can be achieved using the adaptive homodyne measurement scheme detailed in Ref.~\cite{Laverick_2018}. The achievable minimum mean square error (MSE) for such a measurement is given by Eq.~(17) of that work; for the relevant case where the beam phase varies as a Weiner process, this is explicitly given by
\begin{align}\label{MSE_Filt}
    \langle(\phi(T') - \phi_F)^2\rangle \sim \mathfrak{C}^{-1/2}.
\end{align}
In Eq.~(\ref{MSE_Filt}), the asymptotic equivalence holds in the limit of small phase uncertainty, where $|\phi(t) - \phi_{\rm LO}(t)|\ll1$, and $\phi_{\rm LO}(t)$ is the phase of the local oscillator phase used in the adaptive detection scheme. Therefore, for this expression to be used, one must be able to obtain a sufficiently accurate estimate of the beam phase preceding the filtering measurement. In Appendix~\ref{Validity of Linearized Retrofiltering Approximation}, we verify this can be done with a heterodyne measurement provided $\mathfrak{C}\gg1$, which is the relevant regime for us.

Following this initial filtering measurement, we can use Lemma 1 of Ref.~\cite{Baker2020} to make a statement about the state of the device immediately after the filtering measurement. That is, given the outcome $\phi_F$, at time $T'$ the device resides in the fiducial state $\rho_{\rm c}^{\rm fid}$ with an optical phase shift $\phi_F$ encoded by the generator $\hat{n}_{\rm c}$,
\begin{align}\label{state_T'}
    \rho_{{\rm c}|\phi_F}(T') = \mathcal{U}^{\phi_F}_{\rm c}(\rho^{\rm fid}_{\rm c}).
\end{align}
Here, the fiducial state is independent of $\phi_F$. Moreover, by Lemma 3 of Ref.~\cite{Baker2020}, this state shares the same number statistics as the steady state; i.e., $\bra{n}\rho^{\rm ss}_{\rm c}\ket{n} = \bra{n}\rho^{\rm fid}_{\rm c}\ket{n}$. Crucially, this means that the mean excitation number is the same for both states, with ${\rm Tr}(\hat{n}_{\rm c}\rho^{\rm fid}_{\rm c}) = {\rm Tr}(\hat{n}_{\rm c}\rho^{\rm ss}_{\rm c}) = \mu$.

Now, we consider an optimal \textit{retrofiltering} measurement, performed on the beam over the time interval $(T',T'']$, to estimate $\phi_F$. From Ref.~\cite{Laverick_2018}, two statements can be made regarding this optimal retrofiltering measurement: First, it is just as accurate as the optimal filtering estimate; i.e., $V:=\langle(\phi(T')-\phi_R)^2\rangle = \langle(\phi(T')-\phi_F)^2\rangle$. Second, the MSE in the optimal \textit{smoothed} estimate $\phi_S=(\phi_F+\phi_R)/2$ [i.e., the optimal estimate of $\phi(T')$ obtained from both the filtering and retrofiltering information] evaluates to $\langle(\phi(T')-\phi_S)^2\rangle = V/2$. Together, these two statements imply that the covariance $\langle(\phi(T')-\phi_R)(\phi(T')-\phi_F)\rangle$ vanishes. Using this with Eq.~(\ref{MSE_Filt}), we find
\begin{align}\label{MSE_filtret}
    \langle(\phi_F - \phi_R)^2\rangle \sim 2\mathfrak{C}^{-1/2}.
\end{align}

Instead of the optimal retrofiltering measurement detailed above, one could obtain another estimate $\phi_D$ of the encoded phase $\phi_F$ by performing a direct measurement on the laser device at time $T'$. The MSE for this estimate will be as large as the phase variance of the fiducial state $\rho_{\rm c}^{\rm fid}$; i.e.,
\begin{align}\label{direct_variance}
    \langle(\phi_F - \phi_D)^2\rangle = \langle(\Delta\hat{\varphi}_{\rm PB})^2\rangle_{\rm fid}\geq\langle(\Delta\hat{\varphi}_{\rm PB})^2\rangle_{{\rm min},\mu},
\end{align}
where $\langle(\Delta\hat{\varphi}_{\rm PB})^2\rangle_{\rm fid}$ denotes the phase variance of the fiducial state. Here, we have also introduced the quantity $\langle(\Delta\hat{\varphi}_{\rm PB})^2\rangle_{{\rm min},\mu}$, which denotes the smallest phase variance that can be possibly achieved in a quantum state that has a mean excitation number of $\mu$ [recall that ${\rm Tr}(\hat{n}_{\rm c}\rho_{\rm c}^{\rm fid})=\mu$]. Given that the phase information from the initial filtering measurement is encoded onto the fiducial state purely by the generator $\hat{n}_{\rm c}$, the direct measurement must outperform any other measurement that estimates the phase $\phi_F$. This includes the retrofiltering measurement detailed above, which is estimating the cavity phase via the measurement on an intermediate system (that is, the beam segment over the time interval $(T',T'']$). Thus, we have the inequality $\langle(\phi_F - \phi_R)^2\rangle \geq \langle(\phi_F - \phi_D)^2\rangle$; using this with Eqs.~(\ref{MSE_filtret}) and (\ref{direct_variance}) leads to a bound for the coherence in terms of the minimal phase variance
\begin{align}\label{coh_opt_phase}
    \mathfrak{C}\lesssim\frac{4}{\left[\langle(\Delta\hat{\varphi}_{\rm PB})^2\rangle_{{\rm min},\mu}\right]^2}.
\end{align}

From here, we invoke a result from Ref~\cite{Bandilla_1991}, where it was show that
\begin{align}\label{direct_phase}
    \langle(\Delta\hat{\varphi}_{\rm PB})^2\rangle_{{\rm min},\mu} \sim 4\left|\frac{z_A}{3}\right|^3\frac{1}{\mu^2},
\end{align}
in the limit $\mu\gg1$. Substituting Eq.~(\ref{direct_phase}) into Eq.~(\ref{coh_opt_phase}), we obtain an upper bound for the coherence in terms of the energy resource $\mu$,
\begin{align}
    \begin{split}
        \mathfrak{C} & \lesssim \frac{1}{4}\left|\frac{3}{z_A}\right|^6\mu^4 \approx 1.1156\mu^4.
    \end{split}
\end{align}

\end{proof}
As stated above, this upper bound for the coherence is tighter than that derived in Ref.~\cite{Baker2020} by a factor of $3/8$.

\subsection{The $p$-Family of Laser Models}\label{p_fam_revisit}

We conclude this section by providing evidence that a family of laser models exists which satisfies the four Conditions required for Theorem~\ref{thm_1}, and achieves Heisenberg-limited scaling for $\mathfrak{C}$. The family of laser models which achieves this is the so-called ``$p$-family" of laser models, introduced in Refs.~\cite{Ostrowski2022a,Ostrowski2022b}, and was conceived from a straightforward generalization of the original Heisenberg-limited laser model put forth in~\cite{Baker2020}. This family can be realised as a special case of both the $p,\lambda$- and $p,q$-families of laser models by respectively setting $\lambda=0$ and $q=0$. This is the Poissonian case for the two families, where $Q_{t\to\infty,\lambda = 0}=Q_{t\to\infty,q = 0}=0$. The master equation characterizing the $p$-family is
\begin{align}\label{p_master}
    \begin{split}
        \frac{d\rho}{dt} & = \mathcal{L}^{(p)}\rho \\ & = \mathcal{N}\left(\mathcal{D}[\hat{G}^{(p,0)}] + \mathcal{D}[\hat{L}^{(p,0)}]\right)\rho.
    \end{split}
\end{align}

It has already been shown that this family of laser models exhibits a coherence that saturates the Heisenberg-limited scaling of $\mathfrak{C}=\Theta(\mu^4)$ provided that $p=O(1)$ in $\mu$, and further that it satisfies Conditions 1--3 stated above~\cite{Baker2020}. Therefore, we are left to verify Condition 4, the Ideal Laser Beam Condition, which will be shown to hold asymptotically in the linearized regime. 

We first consider the state of the cavity and beam that is produced in an infinitesimal time interval $[t,t+dt)$. From Eq.~(\ref{p_master}), and the assumption that the loss term $\mathcal{N}\mathcal{D}[\hat{L}^{(p,0)}]\rho$ arises from quantum vacuum white noise coupling~\cite{Gardiner2004}, we find
\begin{widetext}
\begin{align}\label{dynamics_long}
    \begin{split}
        {\rm Tr}_{\rm e'}\left[\mathcal{U}_{\rm ce}^{t\to t+dt}(\rho_{\rm c}\otimes\rho_{\rm e})\right] = & \rho_{\rm c}\otimes|0\rangle_{\rm b}\langle0| + \sqrt{\mathcal{N}dt}\left(\hat{L}^{(p,0)}\rho_{\rm c}\otimes|1\rangle_{\rm b}\langle0| + \rho_{\rm c}\hat{L}^{(p,0)\dagger}\otimes|0\rangle_{\rm b}\langle1|\right) + \mathcal{N}dt\mathcal{D}[\hat{G}^{(p,0)}]\rho_{\rm c}\otimes|0\rangle_{\rm b}\langle0| \\ & \hspace{-0.5cm} + \mathcal{N}dt\left(\hat{L}^{(p,0)}\rho_{\rm c}\hat{L}^{(p,0)\dagger}\otimes|1\rangle_{\rm b}\langle1| - \frac{1}{2}\left[\hat{L}^{(p,0)\dagger}\hat{L}^{(p,0)}\rho_{\rm c}+\rho_{\rm c}\hat{L}^{(p,0)\dagger}\hat{L}^{(p,0)}\right]\otimes|0\rangle_{\rm b}\langle0|\right) + o(\mathcal{N}dt).
    \end{split}
\end{align}
\end{widetext}
In this instance, we are using notation where $\rho_{\rm c}\otimes\rho_{\rm e}$ represents the joint state of the cavity (subscript c) and its environment (subscript e). We have also defined the superoperator $\mathcal{U}_{\rm ce}^{t\to t+\tau}:\mathcal{B}(\mathcal{H}_{\rm c}\otimes\mathcal{H}_{\rm e})\xrightarrow{}\mathcal{B}(\mathcal{H}_{\rm c}\otimes\mathcal{H}_{\rm b}\otimes\mathcal{H}_{\rm e'})$ that describes the unitary evolution of the device and its surrounding environment from time $t$ to $t+\tau$. $\mathcal{B}(\mathcal{H})$ denotes the set of bounded linear operators acting on Hilbert space $\mathcal{H}$, and the segment of the beam generated in the time interval $(t,t+\tau]$ is contained within the Hilbert space $\mathcal{H}_{\rm b}$. The primed environment space $\mathcal{H}_{\rm e'}$ contains everything not counted as part of the device and this beam segment; that is, $\mathcal{H}_{\rm e'}$ is a strict factor of $\mathcal{H}_{\rm e}$. Here, we are also treating the state of the beam segment created in $(t,t+dt]$ as belonging to a Fock space truncated after a single excitation. In this space, the beam operators are $\sqrt{dt}\hat{b}(t) = |0\rangle_{\rm b}\langle1|$ and the photon number operator over this small beam segment becomes $dt\hat{b}^\dagger(t)\hat{b}(t)=|1\rangle_{\rm b}\langle1|$.

We consider the cavity to be in the pure state, $\rho_{\rm c}=\varrho^\phi$, defined in Eq.~(\ref{pure_cavity_state}). It turns out that this pure cavity state is essentially an eigenstate of the loss operator:
\begin{align}\label{apx eigen}
    \begin{split}
        \hat{L}^{(p,0)}\varrho^\phi & = \left[\sum_{l=1}^{D-1}\sqrt{\frac{\rho_{l-1}}{\rho_l}}\ketbra{l-1}{l}\right] \\ & \hspace{1cm} \times\left[\sum_{m,n=0}^{D-1}\sqrt{\rho_{m}\rho_{n}}e^{i(m-n)\phi}\ketbra{m}{n}\right] \\ & \hspace{-.5cm} = e^{i\phi}\sum_{m=0}^{D-2}\sum_{n=0}^{D-1}\sqrt{\rho_{m}\rho_{n}}e^{i(m-n)\phi}\ketbra{m}{n} \\ & \hspace{-.5cm} = e^{i\phi}\left(\varrho^\phi +\hat{s}\right),
    \end{split}
\end{align}
where $\hat{s}= -\sum_{n=0}^{D-1}\sqrt{\rho_{D-1}\rho_n}e^{i(D-1-n)\phi}\ketbra{D-1}{n}$. We find that the Hilbert-Schmidt norm of $\hat{s}$ evaluates to
\begin{align}
    \begin{split}
        ||\hat{s}|| = \sqrt{{\rm Tr}(\hat{d}\hat{d}^\dagger)} & = \sqrt{\rho_{D-1}} \\ &  \sim\frac{1}{\pi^{1/4}}\left(\frac{\pi}{2\mu}\right)^{\frac{(p+1)}{2}}\sqrt{\frac{\Gamma\left(\frac{2+p}{2}\right)}{\Gamma\left(\frac{1+p}{2}\right)}},
    \end{split}
\end{align}
where the asymptotic equivalence holds in the limit $\mu\to\infty$. This indicates that the operator $\hat{s}$ is very small relative to $\varrho^\phi$ (where $||\varrho^\phi||=1$) in the linearized regime. Hence, we can safely neglect $\hat{s}$ in the last line of Eq.~(\ref{apx eigen}), and obtain the very accurate approximation $\hat{L}^{(p,0)}\varrho^\phi\approx e^{i\phi}\varrho^\phi$. Under this approximation, the reduced state of the infinitesimal beam segment emitted in $(t,t+dt]$ is in a single mode coherent state,
\begin{align}\label{dynamics_traced}
    {\rm Tr}_{\rm ce'}\left[\mathcal{U}_{\rm ce}^{t\to t+dt}(\varrho^\phi\otimes\rho_{\rm e})\right] \approx |\sqrt{\mathcal{N}dt}e^{i\phi}\rangle_{\rm b}\langle\sqrt{\mathcal{N}dt}e^{i\phi}|,
\end{align}
where we have neglected terms $o(\mathcal{N}dt)$. In Section~\ref{match_cav_diff}, it was shown that the ensemble of cavity states $\{(\varrho^\phi,d\phi/2\pi)\}$ is PR in the regime $1\ll p\ll\mu$, and the dynamics of this ensemble matched that of pure phase diffusion. With this, we can deduce from Eq.~(\ref{dynamics_traced}) that the beam state can be described as a coherent state, also undergoing pure phase diffusion; i.e., the ideal beam state described in Condition 4 (Ideal Laser Beam). 

It was stated that the term \textit{statistically equivalent} used in the description of Condition 4 meant that that the state of beam produced by the laser device matches, to an arbitrarily good approximation, a coherent state undergoing pure phase diffusion. From the analysis presented directly above, we can see that the extent to which the laser beam of the $p$-family of laser models satisfies this condition depends on the extent to which the ensemble of cavity states $\{(\varrho^\phi,d\phi/2\pi)\}$ is PR\footnote{Although the intracavity dynamics for the $p,\lambda$- and $p,q$-families are also describable by pure diffusion for the cases $\lambda\neq0$ and $q\neq0$ in the linearised regime, the pure cavity states $\varrho^\phi$ are no longer eigenstates of their respective loss operators. This means that with these sub-Poissonian families of laser models, it is invalid to assert the beam is in the Ideal Laser Beam State~(\ref{ideal laser}) in general.}. Therefore, the arbitrarily good approximation that we are appealing to here is that of Eq.~(\ref{to_verify}), which essentially states that the evolution of $\varrho^\phi$ converges to pure phase diffusion, with an error scaling as $1/p$ in the limit $\mu\to\infty$ (see Fig.~\ref{fig:diff_verify}a). Because $p$ must be large for this condition to be satisfied, and the coherence decreases with $p$ [see, e.g., Eq.~(\ref{coh_lam_lin}) for $\lambda=0$], the gap between the upper bound~(\ref{upper_bound}) and the coherence of the $p$-family in the linearised regime is greater than the gap between the original upper bound derived in Ref.~\cite{Baker2020} and the $p$-family with the optimized value $p\approx4.15$. Recall also that the fourth condition on the beam used in the theorems of Refs.~\cite{Ostrowski2022a,Baker2020,Ostrowski2022b} are strictly weaker than that used in this section, and do not require large $p$ to be valid.

\section{Conclusions and Outlook}
\label{conclusion}

In this work, we have provided a detailed analysis of the Heisenberg-limited families of laser models introduced in Refs.~\cite{Ostrowski2022a,Ostrowski2022b}. By specializing to a parameter regime that permits a linearization of the operators specifying the gain and loss within these systems, we were able to accurately characterize their dynamical behavior and obtain a number of insights. From this characterization, we found clear analogies with the dynamics that are commonly seen in standard laser models [i.e., models which have a coherence limited by $\mathfrak{C}=O(\mu^2)$], so many of the common intuitions from standard laser theory can be applied to our families, albeit with some key differences.

In the linearised regime, we found that the coherent dynamics of our families of laser models are describable by a PR ensemble of pure states $\{(\varrho^\phi,d\phi/2\pi)\}$, which evolves according to pure diffusion of the phase variable $\phi$. Although these dynamics resemble that of a stochastically evolving coherent state in a standard laser cavity, the cavity state in the Heisenberg-limited case is extremely phase squeezed. Intuitively, this is expected because the dynamics of these systems are specifically designed to preserve phase information within the cavity more effectively than in a standard laser. In addition, we found that the number dynamics of the cavity are describable as an Ornstein-Uhlenbeck process, which is what is also commonly seen in well-known standard laser models~\cite{Haake_1989,Walls_Milburn_2008}.

From our analysis of the dynamics in the linearised regime, we were able to evaluate the corresponding first- and second-order Glauber coherence functions of the laser beams for each family. This led us to obtain formulas for a number of important physical properties of the beam, including their Mandel-$Q$ parameters for arbitrary counting intervals [Eqs.~(\ref{qlam_final})~and~(\ref{qq_final})], the intensity noise spectrum of the beam [Eqs.~(\ref{SIlam_final})~and~(\ref{SIq_final})] and the Coherence [Eqs.~(\ref{coh_lam_lin})~and~(\ref{coh_q_lin})]. All of the analytical results derived here show excellent agreement with numerical simulations in the appropriate parameter regime, where the linearization is valid. It is notable for the sub-Poissonian families that, while the pure cavity state mentioned above was found to be substantially phase squeezed, the output beam is instead number squeezed. This occurs because of the specific nonlinear couplings that characterize the interaction between the laser cavity and its environment.

As a final result, we derived an upper bound for the Coherence of a laser beam. This upper bound exhibits the same $\mu^4$ scaling as in Ref.~\cite{Baker2020}, but the prefactor is tighter by a factor of $3/8$. The general methodology of proof used to derive this bound was the same as in Ref.~\cite{Baker2020}. However, we employed known results regarding optimal filtering and retrofiltering measurements to estimate the phase of the laser beam~\cite{Laverick_2018}, instead of the standard heterodyne measurements used in the original proof. The tighter upper bound accordingly followed from the consideration of these more precise phase estimation procedures. In order to use these known results on optimal filtering and retrofiltering, we required relatively strict constraints to be placed on a laser beam, where its properties are \textit{statistically equivalent} to a coherent state undergoing pure phase diffusion. Although this condition is stricter than those used in previous proofs~\cite{Baker2020,Ostrowski2022b}, we provided strong evidence that this condition is satisfied by a special case of the Heisenberg-limited families of laser models studied in this work.

With our results, we have also reproduced one of the curious observations made in Refs.~\cite{Ostrowski2022a,Ostrowski2022b}.  That is, the so-called ``win-win" relationship between coherence and sub-Poissonianity for Heisenberg-limited laser beams. This is where a positive correlation can be seen between a reduction in the number fluctuations [as quantified by a decrease in the beam's long-time Mandel-$Q$ parameter; see Eqs.~(\ref{qlam_final})~and~(\ref{qq_final})], and a reduction in the phase fluctuations of the beam [as quantified by an increase in the coherence; see Eqs.~(\ref{coh_lam_lin})~and~(\ref{coh_q_lin})]. The exploration of this result has not been a prevalent aim of this paper. However, we believe that the analysis presented here could enable a study into further relations between coherence and number squeezing in laser beams because Eq.~(\ref{g1_general_2}) shows influences from both the phase and number fluctuations on the first-order Glauber coherence function. Phase diffusion dominates the long-time behavior of that function, while the incoherent number dynamics leave a footprint on the short-time behavior. Moreover, given that the dynamics of these linearized Heisenberg-limited models are analogous to those seen with standard lasers, there is reason to believe that such an analysis could encompass a broad range of laser models.

\begin{acknowledgments}
This work was supported by an Australian Research Council Discovery Project (DP220101602) and an Australian Government RTP Scholarship. We acknowledge the support of the Griffith University eResearch Service and Specialised Platforms Team, and the use of the High Performance Computing Cluster “Gowonda” to complete this research.
\end{acknowledgments}

\section*{Data Availability}

The data that support the findings of this article are openly available~\cite{Ostrowski_GitHub}.

\bibliography{apssamp}

\appendix

\section{Modified Laser Pumping Statistics}\label{rp_apx}

In this appendix, we discuss some of the subtleties of the $p,q$-family of laser models and the regime of its validity. The master equation~(\ref{q_master}) describing this model can be derived by considering the repeated action of a CPTP map $\mathcal{G}$ acting on the laser state $\rho$. Here, the map $\mathcal{G}$ is interpreted as a single round of gain, which will raise the laser excitation number by one (${\rm Tr}[\hat{a}^\dagger\hat{a}\mathcal{G}\rho]={\rm Tr}[\hat{a}^\dagger\hat{a}\rho]+1$), and arises from the interaction between the laser mode and a source providing incoherent excitations. Specific instances of this map from particular physical models can be found in Refs.~\cite{Wiseman1993,Wiseman_1999}; at this point however, we do not make any further assumptions about its underlying form. 

The derivation proceeds by considering a time interval $\Delta t$ that is long, such that a large number $n(\Delta t)$ of single gain events have occurred, but short enough that the overall change in the laser state is small; i.e., $||n(\Delta t)(\mathcal{G}-\mathcal{I})\rho||\ll1\ll n(\Delta t)$ ($\mathcal{I}$ is the identity map). This enables the stepwise evolution throughout $\Delta t$ caused by multiple rounds of gain to be treated as a continuous process. Here, $\rho$ is considered to be a ``typical" state obtained throughout the evolution. From our analysis in Section~\ref{match_cav_diff}, it is reasonable consider these as arbitrary mixtures of states $\varrho^\phi$ defined in Eq.~(\ref{pure_cavity_state}). 

If the average rate of excitations added to the cavity is $\mathcal{N}$, and the statistics of $n(\Delta t)$ are characterised by the Mandel-Q parameter $q$, then the coarse-grained evolution caused by this gain process can be described by the time derivative~\cite{Haake_1989,Bergou_1989a,Bergou_1994,Walls_Milburn_2008}
\begin{align}\label{ln_gain}
    \left(\frac{d\rho}{dt}\right)_{\rm gain} = -\frac{\mathcal{N}}{q}{\rm ln}\left[\mathcal{I}-q\left(\mathcal{G}-\mathcal{I}\right)\right]\rho.
\end{align}
To make the analysis of these gain dynamics more tractable, it is common practice to expand the logarithm term under the condition that $||q(\mathcal{G}-\mathcal{I})\rho||\ll1$~\cite{Bergou_1989b,Walls_Milburn_2008}, which gives
\begin{align}\label{ln_gain_2}
    \begin{split}
        \left(\frac{d\rho}{dt}\right)_{\rm gain} = & \mathcal{N}\left((1-q)[\mathcal{G}-\mathcal{I}]+\frac{q}{2}[\mathcal{G}^2-\mathcal{I}]\right)\rho \\ & + O(||(\mathcal{G}-\mathcal{I})^3\rho||).
    \end{split}
\end{align}
For our specific model we set $\mathcal{G}=\hat{G}^{(p,0)}\bullet\hat{G}^{(p,0)\dagger}+\hat{P}\bullet\hat{P}$, with $\hat{P}=\ketbra{D-1}{D-1}$, which gives the gain term in Eq.~(\ref{q_master}) of the main text [for sufficiently large $D$, the second term involving the projector $\hat{P}$ is $O(||(\mathcal{G}-\mathcal{I})^3\rho||)$ for typical states $\rho$ and is therefore neglected from that equation]. Quantum operations similar to this choice for $\mathcal{G}$ have been realised in superconducting circuit experiments; see Ref.~\cite{Gertler_2021}, for example. For a more complete model, one can also incorporate a loss term $(d\rho/dt)_{\rm loss}=\mathcal{L_{\rm loss}\rho}$ for the beam formation, which can be simply added to Eq.~(\ref{ln_gain_2}) provided that the elements of the commutator $||\mathcal{N}[(\mathcal{G}-\mathcal{I})\mathcal{L_{\rm loss}}\rho,\mathcal{L_{\rm loss}}(\mathcal{G}-\mathcal{I})\rho]||$ are negligible compared to the leading terms in Eq.~(\ref{ln_gain_2})~\cite{Bergou_1994}. This yields the full $p,q$-family of laser models provided in the main text. For values $p>3$, the approximations described above become exact in the limit $\mu\to\infty$~\cite{Ostrowski2022a}.

The map $\mathcal{G}^m$ is CPTP; it follows that $(\mathcal{G}^m-\mathcal{I})$ can be expressed as a sum of Lindblad superoperators $\sum_i\mathcal{D}[\hat{K}^{(m)}_i]$, where the jump operators $\hat{K}^{(m)}_i$ are the Kraus operators of $\mathcal{G}^m$. For Poissonian and super-Poissonian pumping sequences, where $q\in[0,1)$, Eq.~(\ref{ln_gain_2}) is manifestly in Lindblad form. In contrast, for sub-Poissonian pumping [$q\in(-1,0)$], the negative coefficient $q/2$ in Eq.~(\ref{ln_gain_2}) prevents the gain dynamics to be expressible in Lindblad form. This issue persists even if all terms in the expansion of the logarithm in Eq.~(\ref{ln_gain}) are considered, and is a caveat that should be noted with all laser models that apply these standard techniques to model a sub-Poissonian pumping process~\cite{Haake_1989,Bergou_1989a,Bergou_1994,Walls_Milburn_2008}. Care must therefore be taken when analyzing results from the $p,q$-family of laser models in the sub-Poissonian pumping regime to ensure the underlying assumptions in its derivation remain valid. In linearized regime, where $1\ll p\ll\mu^2$, these assumptions are indeed well-justified~\cite{Ostrowski2022a}.

\begin{figure}[H]
\includegraphics[width=1.0\columnwidth]{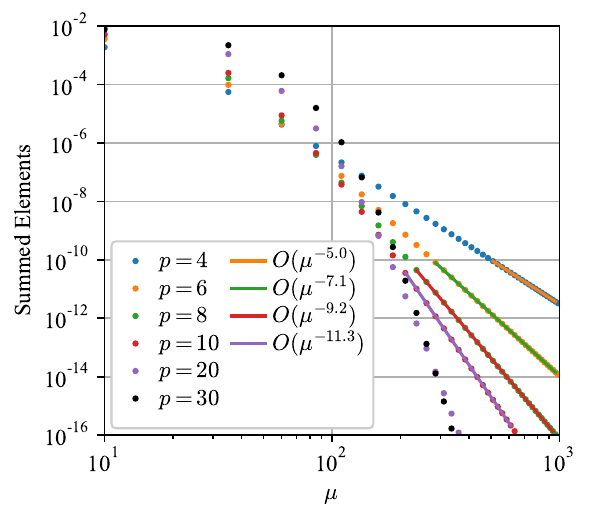}
\caption{\label{fig:apx1}Absolute value of the sum of the negative elements for the steady state distribution of the master equation Eq.~(\ref{q_master}) against $\mu$. Dots are numerical evaluations of this quantity and the solid lines represent power laws fitted to the numerical data below values of $10^{-10}$ for the summed elements, and with $p\in[4,10]$.}
\end{figure}

One consequence of this caveat is the appearance of unphysical effects when a sub-Poissonian pump is considered, such as small negative values in the steady state photon number distribution. These effect become negligible in the parameter regime that validates the approximations underpinning the model are valid. To demonstrate this, we numerically compute the steady-state distribution from Eq.~(\ref{q_master}) in the maximally sub-Poissonian case ($q\to-1$) for a range of values of $p$ and $\mu$. Fig.~\ref{fig:apx1} plots the absolute value of the sum of the negative elements in the resulting distributions as a function of $\mu$, revealing that these unphysical contributions diminish rapidly with $\mu\gg1$. For relatively small $p$, we observe a power-law behaviour, where the summed elements scale as $O(\mu^{p+1})$, with $\mu\gg1$. We expect this trend to persist for larger values of $p$ (e.g., $p=20$, $p=30$), given sufficiently large $\mu$, however rigorous numerical confirmation of this non-essential conjecture was not pursued. 

\section{Cavity Dynamics for the $p,\lambda$- and $p,q$-Families of Laser Models}\label{apendB}

In this Appendix, we provide details into the derivation of Eqs.~(\ref{diff_lam})~and~(\ref{diff_q}). In the linearized regime, these are the approximate expressions for the action of the Liouvillian of the two families of laser models on the pure cavity state $\varrho^\phi$ defined in Eq.~(\ref{pure_cavity_state}). 

We first consider the $p,\lambda$-family of laser models, which can be approximated by Eq.~(\ref{lam_mast_lin}) in the linearized regime. Substituting the pure state $\varrho^\phi$ into that linearized equation, we find
\begin{align}\label{B1}
    \begin{split}
        \frac{\bra{m}\dot\varrho^{\phi}\ket{n}}{\mathcal{N}e^{i\phi(m-n)}} & \\ & \hspace{-1.5cm} = G_{m}(k,\lambda)G_{n}(k,\lambda)\sqrt{\rho_{m-1}\rho_{n-1}} \\ & \hspace{-1cm} -\frac{1}{2}\big([G_{m+1}(k,\lambda)]^2 + [G_{n+1}(k,\lambda)]^2\big)\sqrt{\rho_{m}\rho_{n}} \\ & \hspace{-1cm} + L_{m+1}(k,\lambda)L_{n+1}(k,\lambda)\sqrt{\rho_{m+1}\rho_{n+1}} \\ & \hspace{-1cm} - \frac{1}{2}\left([L_{m}(k,\lambda)]^2 + [L_{n}(k,\lambda)]^2\right)\sqrt{\rho_{m}\rho_{n}}.
    \end{split}
\end{align}
From here, we factorize out the term $\sqrt{\rho_n\rho_m}$ to obtain
\begin{align}\label{B1_fac}
    \begin{split}
        \frac{\bra{m}\dot\varrho^{\phi}\ket{n}}{\mathcal{N}e^{i\phi(m-n)}\sqrt{\rho_{n}\rho_{m}}} & \\ & \hspace{-1.5cm} = G_{m}(k,\lambda)G_{n}(k,\lambda)\sqrt{\frac{\rho_{m-1}\rho_{n-1}}{\rho_{m}\rho_{n}}} \\ & \hspace{-1cm} -\frac{1}{2}\big([G_{m+1}(k,\lambda)]^2 + [G_{n+1}(k,\lambda)]^2\big) \\ & \hspace{-1cm} + L_{m+1}(k,\lambda)L_{n+1}(k,\lambda)\sqrt{\frac{\rho_{m+1}\rho_{n+1}}{{\rho_{m}\rho_{n}}}} \\ & \hspace{-1cm} - \frac{1}{2}\left([L_{m}(k,\lambda)]^2 + [L_{n}(k,\lambda)]^2\right).
    \end{split}
\end{align}
Here, the steady state elements satisfy the equation $\rho_n/\rho_{n-1} = [G_n(k,\lambda)/L_n(k,\lambda)]^2$. Substituting that into Eq.~(\ref{B1_fac}) gives
\begin{align}\raisetag{1\baselineskip}
    \begin{split}
        \frac{\bra{m}\dot\varrho^{\phi}\ket{n}}{\mathcal{N}e^{i\phi(m-n)}} = & -\frac{k^2(2\lambda^2-2\lambda+1)}{8}(m-n)^2\sqrt{\rho_m\rho_n} \\ & + O(k^3).
    \end{split}
\end{align}
Eq.~(\ref{diff_lam}) in the main text readily follows by keeping terms to lowest order in $k$ in the above equation.

A similar procedure is carried out for the $p,q$-family of laser models. Considering Eq.~(\ref{q_mast_lin}) and the pure state $\varrho^\phi$, we find
\begin{align}\raisetag{4\baselineskip}\label{B2}
    \begin{split}
        \frac{\bra{m}\dot\varrho^{\phi}\ket{n}}{\mathcal{N}e^{i\phi(m-n)}} & \\ & \hspace{-1.5cm} = \frac{q}{2}\sqrt{\rho_{m-2}\rho_{n-2}} + (1-q)\sqrt{\rho_{m-1}\rho_{n-1}} \\ & \hspace{-1cm} - (1-q/2)\sqrt{\rho_m\rho_n} \\ & \hspace{-1cm} + L_{m+1}(k,-q/2)L_{n+1}(k,-q/2)\sqrt{\rho_{m+1}\rho_{n+1}} \\ & \hspace{-1cm} - \frac{1}{2}\left([L_{m}(k,-q/2)]^2 + [L_{n}(k,-q/2)]^2\right)\sqrt{\rho_{m}\rho_{n}}.
    \end{split}
\end{align}
We then factorize out the term $\sqrt{\rho_m\rho_n}$, and expand the equation in orders of $k$, as was done for the $p,\lambda$-family. In this case, however, the steady state is characterized by a more complicated equation. That is, to second order in $k$, we have
\begin{align}\label{pq_recur}
    \begin{split}
        \frac{\rho_{n-1}}{\rho_{n}} = & 1 - \frac{k}{2}(1+q/2) + \frac{k^2}{4}(1+q/2)^2\left(n+1/2-\mu\right) \\ & + O(k^3).
    \end{split}
\end{align}
With Eq.~(\ref{pq_recur}), we evaluate~(\ref{B2}), and obtain
\begin{align}
    \begin{split}
        \frac{\bra{m}\dot\varrho^{\phi}\ket{n}}{\mathcal{N}e^{i\phi(m-n)}} = & -\frac{k^2(1+q/2)^2}{8}(m-n)^2\sqrt{\rho_m\rho_n} \\ & + O(k^3).
    \end{split}
\end{align}
Eq.~(\ref{diff_q}) in the main text readily follows by keeping terms to lowest order in $k$ in the above equation.

\section{Equation of Motion for the Off-Diagonal Matrix Elements}
\label{off-diag Append}

In this Appendix, we derive Eq.~(\ref{general OE}) in the main text. This equation approximately describes the dynamics of the off-diagonal matrix elements $\rho_{n,n-1}:=\bra{n}\rho\ket{n-1}$ for the $p,\lambda$- and $p,q$- families of laser models. We begin by considering the respective Master Equations (\ref{lam_mast_lin}) and (\ref{q_mast_lin}) for these two families in the linearized regime,
\begin{subequations}\label{C1}
    \begin{align}\label{C1_a}
        \begin{split}
            \frac{\dot{\rho}_{n,n-1}}{\mathcal{N}} = & G_{n}(k,\lambda)G_{n-1}(k,\lambda)\rho_{n-1,n-2} \\ & -\frac{1}{2}\big([G_{n+1}(k,\lambda)]^2 + [G_{n}(k,\lambda)]^2\big)\rho_{n,n-1} \\ & + L_{n+1}(k,\lambda)L_{n}(k,\lambda)\rho_{n+1,n} \\ & - \frac{1}{2}\left([L_{n}(k,\lambda)]^2 + [L_{n-1}(k,\lambda)]^2\right)\rho_{n,n-1},
        \end{split}
    \end{align}
    \begin{align}\label{C1_b}
        \begin{split}
            \frac{\dot{\rho}_{n,n-1}}{\mathcal{N}} = & \frac{q}{2}\rho_{n-2,n-3} + (1-q)\rho_{n-1,n-2} \\ & - (1-q/2)\rho_{n,n-1} \\ & + L_{n+1}(k,-q/2)L_{n}(k,-q/2)\rho_{n+1,n} \\ & - \frac{1}{2}\big([L_{n}(k,-q/2)]^2 \\ & \hspace{1cm} + [L_{n-1}(k,-q/2)]^2\big)\rho_{n,n-1}.
        \end{split}
    \end{align}
\end{subequations}

We now turn to a continuum description for Eqs.~(\ref{C1_a})~and~(\ref{C1_b}), whereby $n$ is approximated as a continuous variable and the substitutions $\rho_{n,n-1}(t)\to p_{\rm coh}(n,t)$, $G_{n+1}(k,x)\to G(n+1,k,x)=\sqrt{1-kx(n-\mu)}$ and $L_{n}(k,x)\to L(n,k,x)=\sqrt{1+k(1-x)(n-\mu)}$ are made. From these continuum equations, the terms involving $p_{\rm coh}(n-2,t)$ and $p_{\rm coh}(n\pm1,t)$ are expressed in terms of $p_{\rm coh}(n,t)$ by use of a Taylor series expansion, which is truncated at second order. This yields the following two equations
\begin{subequations}
    \begin{align}\label{C2_a}
        \begin{split}
            \frac{\dot{p}_{\rm coh}(n,t)}{\mathcal{N}} \approx & -\frac{1}{2}\Big{\{}\big[ G(n,k,\lambda) - G(n+1,k,\lambda) \big]^2 \\ & \hspace{1cm} + \big[ L(n-1,k,\lambda) - L(n,k,\lambda) \big]^2\Big{\}} \\ & \hspace{1cm} \times {p}_{\rm coh}(n,t) \\ & - \frac{\partial}{\partial n}\Big{\{}\big[G(n,k,\lambda)G(n+1,k,\lambda) \\ & \hspace{1.5cm} - L(n-1,k,\lambda)L(n,k,\lambda)\big] \\ & \hspace{1.25cm} \times{p}_{\rm coh}(n,t)\Big{\}} \\ & + \frac{1}{2}\frac{\partial^2}{\partial n^2}\Big{\{}\big[G(n,k,\lambda)G(n+1,k,\lambda) \\ & \hspace{1.75cm} + L(n-1,k,\lambda)L(n,k,\lambda)\big] \\ & \hspace{1.6cm} \times{p}_{\rm coh}(n,t)\Big{\}},
        \end{split}
    \end{align}
    \begin{align}\label{C2_b}
        \begin{split}
            \frac{\dot{p}_{\rm coh}(n,t)}{\mathcal{N}} \approx & - \frac{1}{2}\big[ L(n-1,k,-q/2)  \\ & \hspace{1cm} - L(n,k,-q/2) \big]^2{p}_{\rm coh}(n,t) \\ & + \frac{\partial}{\partial n}\Big{\{}\big[ L(n-1,k,-q/2)L(n,k,-q/2,\mu) \\ & \hspace{1.5cm} - 1 \big]{p}_{\rm coh}(n,t)\Big{\}} \\ & + \frac{1}{2}\frac{\partial^2}{\partial n^2}\Big{\{}\big[L(n-1,k,-q/2)L(n,k,-q/2) \\ & \hspace{1.75cm} + 1 + q \big]{p}_{\rm coh}(n,t)\Big{\}}.
        \end{split}
    \end{align}
\end{subequations}
Finally, we substitute the explicit expressions for the functions $G(n+1,k,x)$ and $L(n,k,x)$ into (\ref{C2_a}) and (\ref{C2_b}) and keep the leading order terms in $k$ within each derivative. Thus, we find
\begin{subequations}
    \begin{align}\label{C3_a.1}
        \begin{split}
            \frac{\dot p_{\rm coh}(n,t)}{\mathcal{N}} \approx & - (2\lambda^2-2\lambda+1)\frac{k^2}{8}p_{\rm coh}(n,t) \\ & + \frac{\partial}{\partial n}\left[k(n-\mu-1/2)p_{\rm coh}(n,t)\right] \\ & + \frac{\partial^2}{\partial n^2}p_{\rm coh}(n,t),
        \end{split}
    \end{align}
    \begin{align}\label{C3_a.2}
        \begin{split}
            \frac{\dot p_{\rm coh}(n,t)}{\mathcal{N}} \approx & - (1+q/2)^2\frac{k^2}{8}p_{\rm coh}(n,t) \\ & + \frac{\partial}{\partial n}\left[k(1+q/2)(n-\mu-1/2)p_{\rm coh}(n,t)\right] \\ & + (1+q/2)\frac{\partial^2}{\partial n^2}p_{\rm coh}(n,t).
        \end{split}
    \end{align}
\end{subequations}
Both of these equations have the same the general form as that presented in Eq.~(\ref{general OE}) in the main text.

\section{Verification of the Linearized Approximation for Filtering}
\label{Validity of Linearized Retrofiltering Approximation}

In this appendix, we show that the we can legitimately use Eq.~(\ref{MSE_Filt}) for the proof presented in Section~\ref{tighter ub}. This equation corresponds to the MSE in an optimal filtering measurement used to determine the phase $\phi(t)$ of a coherent beam, which varies as a Wiener process. 

In Ref.~\cite{Laverick_2018}, it was shown that an optimal filtering measurement can be achieved with an adaptive homodyne detection scheme. For such a scheme, the detected photocurrent is given by
\begin{align}\label{photocurrent}
    \begin{split}
        I & = 2\sqrt{\mathcal{N}}\sin[\phi(t)-\phi_{\rm LO}(t)] + \zeta(t) \\ & \approx 2\sqrt{\mathcal{N}}[\phi(t)-\phi_{\rm LO}(t)] + \zeta(t),
    \end{split}
\end{align}
where $\mathcal{N}$ is the photon flux of the beam, $\zeta(t)$ is real, classical white noise, satisfying $\langle \zeta(t)\zeta(t') \rangle = \delta(t-t')$ and $\phi_{\rm LO}(t)$ is the phase of a local oscillator, which may be controlled. In moving from the first to second line of Eq.~(\ref{photocurrent}), it is assumed that $|\phi_{\rm LO}-\phi|\ll1$, thus permitting a linearization of the sine function. This linearization allowed the system to be formulated as a linear Gaussian (LG) estimation problem in Ref.~\cite{Laverick_2018}, which led to Eq.~(\ref{MSE_Filt}). For that result to be utilized for our purposes, we must show that a sufficiently precise estimate of the phase can be obtained---for which
the approximation in Eq.~(\ref{photocurrent}) is reasonable---within a timescale that is much shorter that the timescale of the optimal filtering measurement. We will show that this can be achieved with heterodyne detection.

For the case where phase of the beam varies according to a Weiner process, an optimal, unbiased filtering estimator ${\phi}_{\rm F}(t)$ of $\phi(t)$ is the solution to the stochastic differential equation~\cite{Laverick_2018}
\begin{align}\label{filt_estimator}
    d{\phi}_{\rm F}(t) = -4\mathcal{N}V_{\rm F}{\phi}_{\rm F}(t)dt+\sqrt{4\mathcal{N}}V_{\rm F}y(t)dt,
\end{align}
where we have the variance $V_{\rm F}=\langle[{\phi}_{\rm F}(t) - \phi(t)]^2\rangle$ and $y(t):=I(t)+\sqrt{4\mathcal{N}}{\phi}_F(t)$. The stationary variance can be obtained by solving the algebraic matrix Riccati equation~\cite{wiseman_milburn_2009}, yielding $V_{\rm F} = \sqrt{\ell/4\mathcal{N}} = \mathfrak{C}^{-1/2}$, where $\ell$ is the phase diffusion rate of the beam and $\mathfrak{C} = 4\mathcal{N}/\ell$ is its coherence. This is the minimum achievable MSE in a filtering phase estimate of an ideal laser beam~(\ref{ideal laser}), assuming the linearized form of Eq.~(\ref{photocurrent}).

To validate the linearized form of Eq.~(\ref{photocurrent}), and hence the result described above, we first use Eq.~(\ref{filt_estimator}) and $V_{\rm F}$ to define the timescale over which the optimal filtering measurement takes place; i.e., $\tau_{\rm F}:=\sqrt{\mathfrak{C}}/4\mathcal{N}$. Next, we consider a phase estimate $\phi_{\rm het}$ of the coherent beam phase obtained from heterodyne detection. From this estimate, we require that it be obtained within a time $\tau_{\rm het}$ that is much smaller than the filtering timescale,
\begin{align}\label{het_filt_tscale}
    \tau_{\rm het} \ll \tau_{\rm F} = \frac{\sqrt{\mathfrak{C}}}{4\mathcal{N}},
\end{align}
and that it yields a sufficiently small MSE 
\begin{align}
    \langle(\phi-\phi_{\rm het})^2\rangle \ll 1,
\end{align}
so that the linearization in~(\ref{photocurrent}) is reasonable.

It is a known result~\cite{Wiseman_Heterodyne_1997} that the MSE in a phase estimate $\phi_{\rm het}$ of coherent beam with amplitude $\beta$, and a \textit{constant} phase $\phi$ using heterodyne detection is, ideally, 
\begin{align}\label{6B}
    \langle(\phi-\phi_{\rm het})^2\rangle \sim \frac{1}{2|\beta|^2},
\end{align}
in the limit that $|\beta|\to\infty$. With this result, we can place further constraints on $\tau_{\rm het}$. First, because Eq.~(\ref{6B}) holds for a static phase, we require the timescale of the heterodyne measurement $\tau_{\rm het}$ to be much less than the coherence time of the beam $1/\ell$. Second, $\tau_{\rm het}$ must also be large to ensure a sufficiently small MSE $\langle(\phi-\phi_{\rm het})^2\rangle$. Thus, we have the two further conditions
\begin{align}\label{t_het_2}
    \tau_{\rm het} \ll 1/\ell \quad {\rm and} \quad \frac{1}{2}\ll\mathcal{N}\tau_{\rm het}.
\end{align}
If these conditions are met, then $|\beta|^2 \approx \mathcal{N}\tau_{\rm coh}$ and $\langle(\phi-\phi_{\rm het})^2\rangle \approx 1/2\mathcal{N}\tau_{\rm het}\ll1$ [note that with large $\mathfrak{C}$, then $\tau_{\rm het} \ll 1/\ell$ is redundant, given Eq.~(\ref{het_filt_tscale})].

The constrains on $\tau_{\rm het}$ detailed in Eqs.~(\ref{het_filt_tscale}) and (\ref{t_het_2}) can be concisely summarized with the following condition
\begin{align}
    1/2\ll\mathcal{N}\tau_{\rm het}\ll\sqrt{\mathfrak{C}}/4.
\end{align}
Clearly, this constraint can be satisfied if $\mathfrak{C}$ is sufficiently large. In this regime, one can perform a heterodyne measurement on the beam, which yields a sufficiently accurate initial estimate of its phase to validate the linearization in Eq.~(\ref{photocurrent}). This is the relevant regime for us, and therefore justifies our use of the results regarding optimal filtering and retrofiltering phase measurements from Ref.~\cite{Laverick_2018}.

\end{document}